\documentclass[11pt,onecolumn]{article}
\usepackage[utf8]{inputenc}
\usepackage{amssymb}
\usepackage{amsthm}
\usepackage{amsmath}
\usepackage{graphicx}
\usepackage{color,soul}
\usepackage{xcolor}
\usepackage{empheq}
\usepackage[many]{tcolorbox}
\usepackage{mathtools}
\usepackage{makecell}
\usepackage{relsize}
\usepackage{subfig}
\usepackage{stmaryrd}
\usepackage{tikz}
\usepackage{enumitem}
\usepackage{hyperref}
\usepackage{booktabs}
\usepackage[vlined,linesnumbered,ruled,resetcount]{algorithm2e}

\usepackage[left=1.0in,top=1.0in,right=1.0in,bottom=1.0in]{geometry}
\theoremstyle{plain}
\newtheorem{theorem}{Theorem}[section]
\newtheorem{proposition}[theorem]{Proposition}
\newtheorem{definition}[theorem]{Definition}
\newtheorem{lemma}[theorem]{Lemma}

\newtheorem{claim}[theorem]{Claim}

\newtheorem{construction}[theorem]{Construction}

\newcommand{\R}{\mathbb{R}}
\newcommand{\E}{\mathbb{E}}
\newcommand{\ctwo}{c_2}
\newcommand{\cF}{c_F}
\newcommand{\MaxSE}{\operatorname{MaxSE}}
\newcommand{\MeanSE}{\operatorname{MeanSE}}
\newcommand{\diag}{\text{diag}}

\title{Improved Error Bounds for Pure Differentially Private Continual Counting via Matrix Factorization}

\author{
Pavel Arkhipov\thanks{Institute of Science and Technology Austria,
\texttt{pavel.arkhipov@ist.ac.at}}
\hspace{30pt}
Nikita P. Kalinin\thanks{Institute of Science and Technology Austria,
\texttt{nikita.kalinin@ist.ac.at}}
}
\date{}

\usepackage[subfigure]{tocloft}

\begin{document}

\maketitle

\begin{abstract}
Continual counting under pure differential privacy is one of the simplest and most well-studied problems in the continual observation model. Nevertheless, an asymptotic gap remains between the best known upper and lower bounds for maximum squared error and mean squared error: the upper bound is $O(\epsilon^{-2}\log^3 n)$, while the lower bound is $\Omega(\epsilon^{-2}\log^2 n)$, for both error metrics. The best known constant in the upper bound is achieved by the $k$-ary tree mechanism with the subtraction trick, due to Andersson, Pagh, Steiner, and Torkamani (FORC 2025).
In this work, we improve the leading constant in the maximum squared error and the mean squared error. Our approach uses a general matrix factorization mechanism, yielding 
an improved bound for pure-DP continual counting that does not rely on a tree-based construction. The mechanism starts from a good-quality low-dimensional factorization, obtained via gradient-based optimization, and gives an explicit matrix construction that lifts this factorization to arbitrarily large dimensions, further improving its error guarantees. We offer an efficient algorithmic implementation of our mechanism. 
On the lower-bound side, we prove an $\Omega(\epsilon^{-2}\log^3 n)$ lower bound for the class of factorizations whose matrices have entries in $\{0,1\}$, matching the upper-bound asymptotics for this class. 
This class includes the binary tree mechanism and $k$-ary tree mechanisms without the subtraction trick. 
Extending this lower bound to arbitrary matrix factorizations, and beyond the matrix mechanism altogether, remains an open problem.
\end{abstract}
\vspace{1cm}

\begin{center}
\small\bfseries Contents
\end{center}

\vspace{-0.8em}

\tableofcontents
\clearpage

\section{Introduction}

Differential privacy \cite{dwork2006calibrating} has become a standard framework for reasoning about privacy, both in theory and in practice. In many applications, however, data arrive sequentially, and statistics must be released continually as the data stream evolves. This motivates the continual observation model of Dwork, Naor, Pitassi, and Rothblum~\cite{dwork2010differential}, in which the privacy guarantee must hold for the entire sequence of outputs rather than for a single final release. One of the most fundamental problems in this model is private continual counting under pure differential privacy~\cite{dwork2006calibrating,dwork2010differential}. Given a binary stream $x_1,\ldots,x_n \in \{0,1\}$ of sensitive updates, the goal is to release, at each time $t$, an approximation to the prefix sum $\sum_{i=1}^t x_i$, while ensuring that the entire output sequence satisfies $\epsilon$-differential privacy. Continual counting and its variants serve as building blocks in a range of private streaming algorithms, including continual mean estimation~\cite{george2024continual}, quantile estimation \cite{imola2026differentially}, continual histograms~\cite{upadhyay2019sublinear,epasto2023differentially,cardoso2022differentially,henzinger2023differentially}, continual clustering~\cite{la2024making}, graph algorithms~\cite{fichtenberger2021differentially,raskhodnikova2024fully, dinitz2025generalized}, and continual distinct-elements counting~\cite{jain2023counting,andersson2026improved}. Thus, improvements to continual counting can benefit a broad range of higher-level private streaming algorithms.

The classical solution to continual counting is the binary tree mechanism~\cite{dwork2010differential,chan2011private}. In this mechanism, the stream is represented by the leaves of a binary tree, and noisy partial sums are stored at internal nodes. Each update affects only logarithmically many nodes, which controls the privacy cost, and each prefix sum can be reconstructed from logarithmically many noisy measurements, which controls the error. Combined with the Laplace mechanism, the binary tree construction gives maximum expected squared error at most $2\log_2^3(n)/\epsilon^2$ for streams of length $n$.
Since no asymptotic improvement over this upper bound is known for pure differential privacy, the leading constant has become an important measure of progress.
The subtraction trick of Honaker~\cite{honaker2015efficient} reduces the error by allowing a prefix sum to be represented using complementary intervals. Further improvements are obtained by replacing the binary tree with a $k$-ary tree~\cite{qardaji_2013,cardoso2022differentially}, thereby optimizing the tradeoff between privacy cost and reconstruction error. Combining these ideas, Andersson, Pagh, Steiner, and Torkamani~\cite{andersson2024count} obtained the best known leading constants for pure-DP continual counting, reducing the error by more than a factor of $8$ compared to the original binary tree mechanism. The construction can be further refined by averaging several estimators; we refer to the resulting method as the averaged $k$-ary tree with subtraction, which increases the improvement over the binary tree mechanism from a factor of $8$ to a factor of $17$.

The tree mechanisms above can all be viewed through the more general matrix mechanism framework~\cite{li2015matrix}. In this framework, the mechanism first takes linear measurements of the stream, adds Laplace noise scaled to their sensitivity, and then uses these noisy measurements to linearly reconstruct the prefix sums. Matrix factorizations have been useful in the approximate DP setting with Gaussian noise \cite{pillutla2025correlated}. For pure DP continual counting, however, they had not previously led to improvements beyond tree constructions. In this work, we show that a general matrix factorization can in fact improve the best-known bounds for pure-DP continual counting, reducing the error by a factor $25$ over the binary tree mechanism.

The best known lower bounds are substantially smaller than the best known upper bounds: for both mean squared error~\cite{henzinger2023almost} and maximum squared error~\cite{dwork2010differential,edmonds2020power}, the strongest known lower bounds are of order $\Omega(\epsilon^{-2}\log^2 n)$, with analogous bounds holding even under approximate differential privacy. Thus, the asymptotic gap between upper and lower bounds remains open. Nevertheless, the lack of asymptotically better mechanisms has motivated the view that the current upper bound may be optimal~\cite{andersson2024count}. In this work, we give partial evidence in this direction: we prove an $\Omega(\epsilon^{-2}\log^3 n)$ lower bound for the class of matrix factorizations whose factors have all entries in $\{0,1\}$. This class is rich enough to include the binary tree mechanism and the $k$-ary tree mechanisms without the subtraction trick. Hence, within this natural class of tree-like factorizations, the known upper bound is asymptotically tight.

\subsection{Problem Statement}
We consider continual counting over a stream
$x=(x_1,\ldots,x_n) \in \{0,1\}^n$.
At each time $t \in [n]$, the goal is to release a private estimate of
$\sum_{j=1}^t x_j$. Equivalently, the target output is $T_n x$, where
$T_n \in \R^{n \times n}$ is the prefix-sum matrix defined by
$(T_n)_{tj} = \mathbf{1}\{j \le t\}$ for $t,j \in [n]$.
Let $\mathcal M$ be a randomized mechanism whose output
$\mathcal M(x) \in \R^n$ contains the estimate $\mathcal M(x)_t$ at time $t$.
We assume that $\mathcal M$ is $\epsilon$-differentially private with respect to
neighboring streams, where $x,x' \in \{0,1\}^n$ are neighbors if
$\|x-x'\|_1 \le 1$.
To measure the utility of the mechanism, we use the maximum squared error and
the mean squared error, defined as follows:
\begin{align}
    &\MaxSE(\mathcal M,n)
    :=
    \max_{x \in \{0,1\}^n}
    \max_{t \in [n]}
    \E_{\mathcal M}
    \left[
        \left(
            \mathcal M(x)_t - (T_nx)_t
        \right)^2
    \right],\\
    &\MeanSE(\mathcal M,n)
    :=
    \max_{x \in \{0,1\}^n}
    \E_{\mathcal M}
    \left[
        \frac{1}{n}
        \left\|
            \mathcal M(x) - T_nx
        \right\|_2^2
    \right].
\end{align}
We formulate the differentially private continual counting problem as finding
an $\epsilon$-DP mechanism $\mathcal{M}$ that minimizes these errors.
In particular, to establish upper bounds on these errors, we consider a
specific family of mechanisms, namely matrix factorization mechanisms.
Consider a factorization $T_n = LR$, where
$L \in \R^{n \times r}$ and $R \in \R^{r \times n}$. To make the product $T_n x$ private, we first multiply $x$ by $R$, add Laplace noise to the resulting vector $Rx$, and then postprocess the noisy vector by multiplying it by $L$. The associated
 matrix mechanism is
\begin{equation}
    \mathcal M_{L,R}(x)
    :=
    L(Rx + z),
\end{equation}
where $z \in \R^r$ has independent Laplace coordinates with scale
$\max_{i \in [n]}\|R_{:,i}\|_{1}/\epsilon$, i.e., the maximum $\ell_1$ norm
of a column of $R$ divided by $\epsilon$. For this mechanism, the errors can be computed as
\begin{equation}
    \MaxSE(\mathcal M_{L,R},n)
    =
    \frac{2}{\epsilon^2}
    \|L\|_{2 \to \infty}^2
    \|R\|_{1 \to 1}^2,
    \qquad
    \MeanSE(\mathcal M_{L,R},n)
    =
    \frac{2}{\epsilon^2}
    \frac{1}{n}
    \|L\|_F^2
    \|R\|_{1 \to 1}^2,
\end{equation}
where $\|\cdot\|_{2 \to \infty}$ denotes the maximum $\ell_2$ norm of the rows,
and $\|\cdot\|_{1 \to 1}$ denotes the maximum $\ell_1$ norm of the columns. To simplify notation, we introduce max-error and mean-error
\emph{factorization costs}\footnote{We call $\ctwo(M)$ and $\cF(M)$
\emph{factorization costs} rather than norms, since we do not claim that they
satisfy a triangle inequality.}
which can be defined for a general matrix $M \in \R^{n \times n}$ as follows:
\begin{align}
    &\ctwo(M)
    :=
    \inf_{\substack{
        r \ge 1,\,
        L \in \R^{n \times r},\,
        R \in \R^{r \times n} \\
        LR = M
    }}
    \|L\|_{2 \to \infty}\,\|R\|_{1 \to 1},\\
    &\cF(M)
    :=
    \inf_{\substack{
        r \ge 1,\,
        L \in \R^{n \times r},\,
        R \in \R^{r \times n} \\
        LR = M
    }}
    \frac{1}{\sqrt n}\,
    \|L\|_F\,\|R\|_{1 \to 1}.
\end{align}

Consequently, optimizing over all factorizations $T_n = LR$ reduces the problem
to bounding the factorization costs $\ctwo(T_n)$ and $\cF(T_n)$.

\begin{table}[t]
\caption{Comparison of our factorization-based mechanism with prior mechanisms for MaxSE and MeanSE. The upper bounds are stated up to a multiplicative factor of $1+o(1)$. For the $k$-ary trees with and without the subtraction trick, we substitute the best value of $k$ for the corresponding metric. The proposed factorization-based mechanism uses two different factorizations for MaxSE and MeanSE, respectively. The averaged $k$-ary tree with subtraction is described in Subsection~\ref{sub:avaraged_k_ary}. }
\centering
\small
\begingroup
\renewcommand{\arraystretch}{2.5}
\begin{tabular}{lcc}
\toprule
Method & MaxSE & MeanSE \\
\midrule
Binary tree \cite{dwork2010differential}
& $\displaystyle 2\frac{\log_2(n)^3}{\epsilon^2}$
& $\displaystyle \frac{\log_2(n)^3}{\epsilon^2}$
\\
Smooth binary tree \cite{andersson2023smooth}
& $\displaystyle 0.25 \frac{\log_2(n)^3}{\epsilon^2}$
& $\displaystyle 0.25 \frac{\log_2(n)^3}{\epsilon^2}$
\\

$k$-ary tree \cite{cardoso2022differentially, qardaji_2013}
& $\displaystyle \underbrace{\frac{32}{\log_2(17)^3}}_{\approx 0.4686}\frac{\log_2(n)^3}{\epsilon^2}$
& $\displaystyle \underbrace{\frac{16}{\log_2(17)^3}}_{\approx 0.2343}\frac{\log_2(n)^3}{\epsilon^2}$
\\

$k$-ary tree + subtraction \cite{andersson2024count}
& $\displaystyle \underbrace{\frac{16}{\log_2(17)^3}}_{\approx 0.2343}\frac{\log_2(n)^3}{\epsilon^2}$
& $\displaystyle \underbrace{\frac{180}{19\log_2(19)^3}}_{\approx 0.1236}\frac{\log_2(n)^3}{\epsilon^2}$
\\
averaged $k$-ary tree + subtraction
& $\displaystyle \underbrace{\frac{8}{ \log_2(17)^3}}_{\approx 0.1171}\frac{\log_2(n)^3}{\epsilon^2}$
& $\displaystyle \underbrace{\frac{8}{ \log_2(17)^3}}_{\approx 0.1171}\frac{\log_2(n)^3}{\epsilon^2}$
\\

Factorization Based (ours) 
& $\displaystyle 0.0778 \cdot\frac{\log_2(n)^3}{\epsilon^2}$
& $\displaystyle 0.0710\cdot\frac{\log_2(n)^3}{\epsilon^2}$
\\

\bottomrule
\end{tabular}
\label{tab:tree-methods}
\endgroup
\end{table}

\subsection{Our contribution}

Our main contribution is a new upper bound for pure-DP continual counting based on general matrix factorization.

Our construction starts from a good-quality factorization
$AB=T_k$ for some fixed $k$, found numerically by gradient-based optimization, and lifts it to
arbitrarily large dimensions by an explicit recursive operation on factorizations.
More precisely, given factorizations $AB = T_n$ and $CD = T_k$, we can construct a factorization $LR = T_{2kn+n+k}$ with $c_2$ and $c_F$ costs satisfying
\begin{equation}
\begin{split}
        \left( \|L\|_{2 \to \infty} \|R\|_{1 \to 1} \right)^{2/3} &\leq
    \left( \|A\|_{2 \to \infty} \|B\|_{1 \to 1} \right)^{2/3} +
    \left( \|C\|_{2 \to \infty} \|D\|_{1 \to 1} \right)^{2/3}, \text{ or} \\
    \left(\frac{1}{\sqrt{2kn+n+k}}\|L\|_{F} \|R\|_{1 \to 1} \right)^{2/3} &\leq
    \left(\frac{1}{\sqrt{n}} \|A\|_{F} \|B\|_{1 \to 1} \right)^{2/3} +
    \left(\frac{1}{\sqrt{k}} \|C\|_{F} \|D\|_{1 \to 1} \right)^{2/3},
\end{split}
\end{equation}
depending on which cost we are interested in. Iterating this operation gives asymptotically good factorizations of $T_n$ for all
$n$.

\begin{theorem}[Main upper bound]
\label{thm:main_theorem}
For every stream length $n$, there exist two $\epsilon$-differentially private continual
counting mechanisms satisfying respectively
\[
    \MaxSE(\mathcal{M},n)
    \leq
    0.0778 
    \frac{\log_2^3 n}{\epsilon^2}(1 + o(1)), \qquad\MeanSE(\mathcal{M},n)
    \leq
    0.0710
    \frac{\log_2^3 n}{\epsilon^2} (1 + o(1)).
\]
\end{theorem}

These constants improve on the previous best pure-DP bounds of the $k$-ary tree
mechanism with the subtraction trick.  In particular, the leading constant for
maximum squared error improves by a factor of about $3$ in comparison to \cite{andersson2024count} and about $1.5$ in comparison to the averaged $k$-ary tree mechanism with subtraction. The leading constant
for mean squared error improves by a factor of about $1.6$. The comparisons with other methods are given in Table~\ref{tab:tree-methods}. Some of the errors had not previously been computed in the pure-DP setting, we provide the missing entries of the table in Appendix~\ref{appendix_table_analysis}.

We also show that the resulting mechanism can be implemented efficiently. Although
the construction is described as a matrix factorization, the matrices obtained by the
recursive construction are sparse and structured.  This allows the correlated noise
term $Lz$ to be generated online, yielding an implementation with
$O_{A,B}(\log n)$ space and $O_{A,B}(n\log n)$ total running time for any fixed base
factorization $AB=T_k$.

On the lower-bound side, we prove that the $\Omega(\log^3 n)$ squared-error behavior is
unavoidable for any matrix-factorization-based mechanism with factors being $0/1$ matrices. Specifically, if
$L$ and $R$ have entries in $\{0,1\}$ and satisfy $LR=T_n$, then
\begin{equation}
    \begin{split}
        \|L\|_{2\to\infty}\|R\|_{1\to 1}
    &\geq
    0.290\,\log_2^{3/2}(n-1) = \Omega\left(\log^{3/2}n\right), \text{ and} \\
    \frac{1}{\sqrt n}\|L\|_F\|R\|_{1\to 1}
    &\geq
    \left(\frac{\log 2}{6}\right)^{3/2}
    \sqrt{1-\frac1n}\,\log_2^{3/2}(n-1) = \Omega\left(\log^{3/2}n\right).
    \end{split}
\end{equation}
We prove this by reducing to the ``skew Bollob\'as'' inequalities from extremal combinatorics. The main idea is to interpret the matrices as encodings of collections of sets satisfying certain pairwise intersection properties. Consequently, every pure-DP matrix mechanism arising from such a $0/1$ factorization
has maximum and mean squared error at least $\Omega(\epsilon^{-2}\log^3 n)$.
This class includes the binary tree mechanism and the standard $k$-ary tree
mechanisms without the subtraction trick. We point out that for $\MaxSE$, our bound yields $\MaxSE(\mathcal M_{L,R},n) \geq 0.168 \frac{\log_2^3 n}{\epsilon^2} (1+o(1))$ for any $0/1$ factorization $LR = T_n$, with the constant $0.168$ being greater than the constant in our upper bound. Therefore, to further improve the leading constant for $\MaxSE$, one cannot remain in the class of $0/1$ factorizations. Previously known lower bounds for the factorization costs were of order $\Omega(\log n)$.

Finally, we establish some structural facts about these factorization costs.  In
particular, we show that $c_2$ and $c_F$ are not matrix norms, since they fail the
triangle inequality. This helps motivate why the pure-DP setting is not as well-behaved as the approximate-DP
setting, where the analogous $\ell_2$-sensitivity factorization costs are matrix norms.

Our upper bound theorem depends on specific base factorizations. We provide both base factorizations, as well as code for verifying the claimed factorization norms using exact rational (fraction-based) arithmetic, via the following link. 

\url{https://github.com/npkalinin/Pure_DP_Continual_Counting_with_MF}

\subsection{Related Work}

\paragraph{Upper bounds.}

The first method for continual counting, and also the first use of the continual observation model, was the binary tree mechanism \cite{dwork2010differential, chan2011private}. Omitting lower-order terms, it gives MaxSE $2\epsilon^{-2}\log_2(n)^3$ and MeanSE  $\epsilon^{-2}\log_2(n)^3$. There are several ways to improve it. First, one can generalize the binary tree mechanism to a $k$-ary tree, in which each node has $k$ children~\cite{qardaji_2013, cardoso2022differentially}. Optimizing over the choice of $k$ decreases the error. Honaker \cite{honaker2015efficient} proposed using a subtraction trick, which changes how the prefix sum is decomposed into a sum of subintervals. This work also discussed the idea of using multiple weighted measurements to construct an estimator, which was later generalized to computing an optimal decoding matrix $L$ based on the pseudoinverse of $R$. Andersson and Pagh \cite{andersson2023smooth} proposed a smooth binary tree mechanism that addresses the discrepancy between MaxSE and MeanSE, reducing the MaxSE. Combining the subtraction trick with $k$-ary trees, Andersson, Pagh, Steiner, and Torkamani \cite{andersson2024count} achieved the best known constants for the problem. Finally, their construction was further refined in the implementation code for the paper by Imola, Boninsegna, Keller, Aamand, Chowdhury, and Pagh~\cite{imola2026differentially} on private quantile estimation. In particular, the code exploits the idea of averaging multiple decompositions of the prefix interval, empirically improving the constants. However, this method has not been formally described elsewhere. With the authors' permission, we provide a description and error analysis for this method, which we refer to as the averaged $k$-ary tree with subtraction; see Subsection~\ref{sub:avaraged_k_ary}.

Matrix factorization has been proposed for general continual linear queries, with arbitrary weighted sums of the input stream, by Li, Miklau, Hay, McGregor, and Rastogi \cite{li2015matrix}. The optimal factorization was formulated as a semidefinite program with rank constraints, but without an upper bound on the running time. Numerically, we found it impractical to use, and instead used gradient optimization to find high-quality finite-dimensional factorizations.

We summarize the constants in Table~\ref{tab:tree-methods}.
The proofs of the error terms for the related works are scattered across the literature: the MeanSE for the $k$-ary trees with and without subtraction was computed explicitly in \cite{andersson2024count}, while the maximum squared error for the $k$-ary trees with subtraction was computed by Andersson, Jain, and Sivakumar \cite{andersson2026improved} in Theorem 9.16, taking $D=k=1$ in their notation. For completeness, we present the analysis of the missing entries in Appendix~\ref{appendix_table_analysis}.

\paragraph{Lower bounds.}
The best-known lower bounds for continual counting are of order $\Omega(\log^2(n)/\epsilon^2)$, both for the mean squared error \cite{henzinger2023almost} and for the maximum squared error \cite{dwork2010differential, edmonds2020power}. Similar bounds hold under approximate differential privacy as well. 
In this work, we establish $\Omega(\log^{3}(n)/\epsilon^{3})$ lower bounds in a specific setting  of factorizations with $0/1$ entries, using combinatorial techniques based on Bollob\'as-type inequalities~\cite{skew_bollobas}. To the best of our knowledge, this approach has not previously been used either in the context of matrix factorizations or  generally in the differential privacy literature.

In concurrent work, Bairaktari and Larsen~\cite{bairaktari2026binarytreemechanismoptimal} prove
a lower bound for pure DP continual counting of $\Omega(\epsilon^{-2}\log^3(n))$ for the squared $\ell_{\infty}$ metric. This metric is related to the metrics that we consider, $\MaxSE$ and $\MeanSE$. In general, being squared, the $\ell_{\infty}$ metric upper bounds $\MaxSE$ and $\MeanSE$, and therefore their bounds do not directly transfer to the settings considered here. Nevertheless, their result gives hope that a similar lower bound may eventually be proved for $\MaxSE$ and $\MeanSE$. 

\paragraph{Continual counting with approximate DP.}
The problem under approximate differential privacy has received, perhaps, even more attention in recent years due to its applications to model training via DP-SGD \cite{kairouz2021practical, Denisov, choquette2022multi, kalinin2024, kalinin2025back, kalinin2025learningrate, kalinin2026dplambdacgd, kalinin2026beyond}. For approximate DP, one commonly uses Gaussian noise for matrix factorization, which results in $\ell_2$ sensitivity rather than $\ell_1$ sensitivity. This setting is more theoretically amenable, since the costs corresponding to our $c_2$ and $c_F$, namely $\gamma_2$ and $\gamma_F$, obtained by replacing the $\|.\|_{1\to 1}$ norm with the $\|.\|_{1\to 2}$ norm, are in fact matrix norms. For these norms, the asymptotics are known and tight; moreover, the leading asymptotic term is known with the constant \cite{fichtenberger2022constant}. Several works have focused on refining the next asymptotic term \cite{henzinger2025improved}, with the best bounds achieved by Henzinger, Kalinin, and Upadhyay \cite{henzinger2025normalized}.

The recent intensive progress in the approximate DP setting motivates the study of constants in the pure DP setting, which has proved to be more challenging. We remark that the constructions used to optimize the $\gamma_2$ and $\gamma_F$ norms would give the wrong asymptotics if applied to the pure DP case, thus calling for fundamentally new constructions.

\subsection{Preliminaries}

\paragraph{Notation.}
For a matrix $A \in \R^{m \times n}$ and $p,q \in [1,\infty]$, we write
$\|A\|_{p \to q} := \sup_{\|x\|_p \le 1} \|Ax\|_q$ for the induced operator
norm from $\ell_p^n$ to $\ell_q^m$. In particular,
$\|A\|_{2 \to \infty} = \max_{i \in [m]} \|A_{i,:}\|_2$ and
$\|A\|_{1 \to 1} = \max_{j \in [n]} \|A_{:,j}\|_1$, where $A_{i,:}$ and
$A_{:,j}$ denote the $i$-th row and $j$-th column of $A$, respectively. We also
write $\|A\|_F := (\sum_{i,j} A_{ij}^2)^{1/2}$ for the Frobenius norm. We use $\log(.)$ for the natural logarithm and $\log_2(.)$ for the binary logarithm.
For two matrices \(A \in \mathbb{R}^{m \times n}\) and \(B \in \mathbb{R}^{p \times q}\), the Kronecker product of \(A\) and \(B\), denoted by \(A \otimes B\), is the
matrix in \(\mathbb{R}^{mp \times nq}\) defined by
\begin{equation}
A \otimes B =
\begin{pmatrix}
a_{11}B & a_{12}B & \cdots & a_{1n}B \\
a_{21}B & a_{22}B & \cdots & a_{2n}B \\
\vdots  & \vdots  & \ddots & \vdots  \\
a_{m1}B & a_{m2}B & \cdots & a_{mn}B
\end{pmatrix}.
\end{equation}

For the $d \times d$ identity matrix, we write $I_d$. If the dimension is clear from the context, we may use just $I$. The all-ones $m \times k$ matrix is denoted as $J_{m \times k}$, and the subscript is omitted if it is clear. We use the notation $\mathbf{1}_n = J_{n \times 1}$ and $\mathbf{1}_n^T = J_{1 \times n}$. 
Given a collection of matrices $M_1, \ldots, M_k$ that are not necessarily square, we denote the block-diagonal matrix formed with $M_i$ as $\diag(M_1, \ldots, M_k)$.

\paragraph{Differential privacy.}
We use the standard notion of pure differential privacy and the Laplace mechanism of Dwork, McSherry, Nissim, and Smith~\cite{dwork2006calibrating}. Let $\mathcal{X}$ be a data universe, and say that two datasets $D,D' \in \mathcal{X}^n$ are neighboring, written $D \sim D'$, if they differ in the data of at most one individual. A randomized mechanism $\mathcal{M} : \mathcal{X}^n \to \mathcal{Y}$ is $\epsilon$-differentially private if, for every $D \sim D'$ and every measurable $S \subseteq \mathcal{Y}$,
\begin{equation}
\Pr[\mathcal{M}(D) \in S]
\le
e^\epsilon \Pr[\mathcal{M}(D') \in S].
\end{equation}
For a query $f : \mathcal{X}^n \to \mathbb{R}^k$, its $\ell_1$-sensitivity is $\Delta_1(f) := \sup_{D \sim D'} \|f(D)-f(D')\|_1$. The Laplace mechanism outputs $f(D)+z$, where the coordinates of $z$ are independent $\mathrm{Laplace}(\Delta_1(f)/\epsilon)$ random variables, and the resulting mechanism is $\epsilon$-differentially private.

\section{Upper bound construction}

We start with two inequalities that are going to be fundamental in this paper.
\begin{theorem}
    \label{th_fundamental_inequalities}
    For any positive integers $n$, $k$, we have
    \begin{equation}
        \begin{split}
            &c_2(T_{2kn + n + k})^{2/3} \leq c_2(T_n)^{2/3} + c_2(T_k)^{2/3}, \\
            &c_F(T_{2kn + n + k})^{2/3} \leq c_F(T_n)^{2/3} + c_F(T_k)^{2/3}.
        \end{split}
    \end{equation}
\end{theorem}

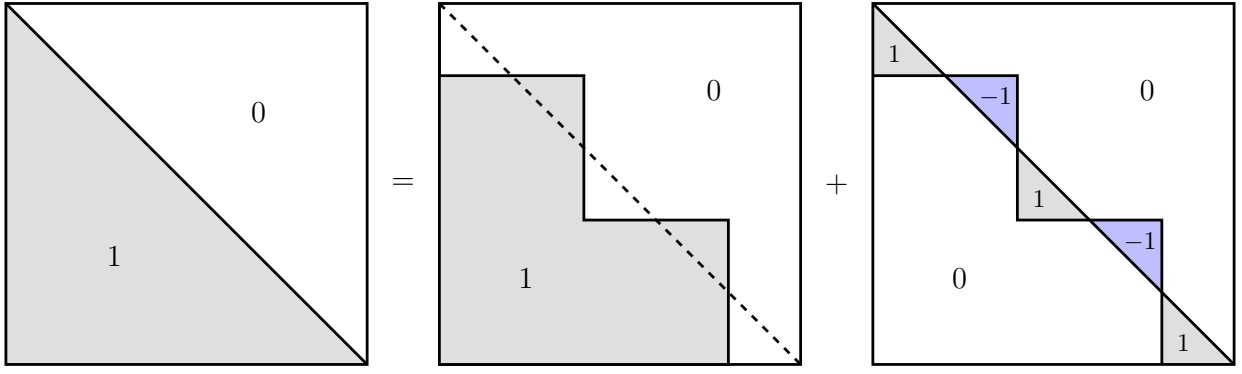
\begin{figure}[!h]
    \centering
    \resizebox{1.\textwidth}{!}{
\begin{tikzpicture}[
    line/.style={draw=black, line width=1.1pt},
    mat/.style={line, minimum width=3.2cm, minimum height=3.2cm},
    font=\large
]

\begin{scope}[shift={(0,0)}]
    \fill[gray!25] (0,0) -- (5,0) -- (0,5) -- cycle;
    \draw[line] (0,0) rectangle (5,5);
    \draw[line] (0,5) -- (5,0);

    \node at (1.5,1.5) {$1$};
    \node at (3.5,3.5) {$0$};
\end{scope}

\node at (5.5,2.5) {$=$};

\begin{scope}[shift={(6,0)}]
    \draw[line] (0,0) rectangle (5,5);

    \fill[gray!25] (0,0) rectangle (2,4);
    \fill[gray!25] (0,0) rectangle (4,2);

    \draw[line] (0,4) -- (2,4) -- (2,2) -- (4,2) -- (4,0) -- (0,0) -- (0,5);

    \draw[line, dashed] (0,5) -- (5,0);

    \node at (1.2,1.2) {$1$};
    \node at (3.8,3.8) {$0$};
\end{scope}

\node at (11.5,2.5) {$+$};

\begin{scope}[shift={(12,0)}]
    \fill[gray!25] (0,4) -- (1,4) -- (0,5) -- cycle;
    \fill[blue!25] (2,4) -- (1,4) -- (2,3) -- cycle;
    \fill[gray!25] (2,2) -- (3,2) -- (2,3) -- cycle;
    \fill[blue!25] (4,2) -- (3,2) -- (4,1) -- cycle;
    \fill[gray!25] (4,0) -- (5,0) -- (4,1) -- cycle;
    
    \draw[line] (0,0) rectangle (5,5);

    \draw[line] (0,5) -- (5,0);

    \draw[line] (0,4) -- (2,4) -- (2,2) -- (4,2) -- (4,0);
    
    \node at (0.3,4.3) {\small $1$};
    \node at (1.7,3.7) {\small$ -1$};
    \node at (2.3,2.3) {\small$ 1$};
    \node at (3.7,1.7) {\small$ -1$};
    \node at (4.3,0.3) {\small$ 1$};

    \node at (1.2,1.2) {$0$};
    \node at (3.8,3.8) {$0$};
\end{scope}

\end{tikzpicture}
}
    \caption{Decomposition of $T_{2kn+n+k}$ into two matrices, example for $k=2$.}
    \label{fig_tn_decomposition}
\end{figure}

To establish these inequalities, we will describe a way to construct a factorization of $T_{2kn + n + k}$ with low cost given factorizations of $T_n$ and $T_k$. The definition below is designed with Figure~\ref{fig_tn_decomposition} in mind.

\begin{definition}
    Given a positive parameter $\alpha$, we define an operation $\lozenge_\alpha$: given two pairs of matrices $(A, B)$ and $(C, D)$, where $A \in \R^{n \times r}$, $B \in \R^{r \times n}$, $C \in \R^{k \times s}$, $D \in \R^{s \times k}$, we have $(A, B) \lozenge_\alpha (C, D) = (L, R)$, where $L \in \R^{(2kn+n+k) \times (2kr+r+s)}$ and $R \in \R^{(2kr+r+s) \times (2kn+n+k)}$ are matrices constructed as follows. 

    First, we introduce matrices $\widetilde{A} \in \mathbb{R}^{(2n+1) \times 2r}$ and $\widetilde{B} \in \mathbb{R}^{2r \times (2n+1)}$, defined as 
    \begin{equation}
        \widetilde{A} = \begin{pmatrix}
            -EA & 0_{n \times r} \\ 0_{1 \times r} & 0_{1 \times r} \\ 0_{n \times r} & A
        \end{pmatrix},\quad
        \widetilde{B} = \begin{pmatrix}
            B & 0_{r \times 1} & 0_{r \times n} \\ 0_{r \times n} & 0_{r \times 1} & BE
        \end{pmatrix},
    \end{equation}
    where $E$ is an $n \times n$ matrix with ones on the anti-diagonal: $E_{ij} = \mathbf{1}\{i+j = n+1\}$. Then we define
    \begin{equation}
        L = \begin{pmatrix}
            A & 0_{n \times 2rk} & 0_{n \times s}\\
            0_{k(2n+1) \times r} & I_{k} \otimes \widetilde{A} & \alpha C \otimes \mathbf{1}_{2n+1}
        \end{pmatrix},
        \quad
        R = \begin{pmatrix}
            I_k \otimes \widetilde{B} & 0_{2rk \times n} \\
            0_{r \times k(2n+1)} & B \\
            \frac{1}{\alpha} D \otimes \mathbf{1}_{2n+1}^T & 0_{s \times n}
        \end{pmatrix}.
    \end{equation}
\end{definition}

Notice that the matrix $EA$ is $A$ with rows in inverse order, and $BE$ is $B$ with columns in inverse order. The definition above may look somewhat complicated, but it makes sense in the context of matrix factorization, since it allows us to factorize $T_{2kn+n+k}$ given factorizations for $T_n$ and $T_k$, as shown in the following lemma.
\begin{lemma}
    Let $AB = T_n$ and $CD = T_k$, and let $(L, R) := (A, B) \lozenge_\alpha (C, D)$ for any positive $\alpha$. Then $LR = T_{2kn+n+k}$.
\end{lemma}
\begin{proof}
    Let us denote $U_m = J_{m \times m} - T_m$, so $(U_m)_{ij} = \mathbf{1}\{i < j\}$.
    Observe that 
    \begin{equation}
        \label{eq_t_decomposition}
        T_{2kn + n + k} = 
        \begin{pmatrix}
            0_{n \times k(2n+1)} & 0_{n \times n} \\
            T_k \otimes J_{(2n+1) \times (2n+1)} & 0_{k(2n+1) \times n}
        \end{pmatrix} +
        \diag(\underbrace{T_n, -U_{n+1}, T_n, -U_{n+1}, \ldots, T_n}_{2k+1 \text{ blocks}}).
    \end{equation}
    This equality is illustrated in Figure~\ref{fig_tn_decomposition}.

    The first term on the right-hand side of \eqref{eq_t_decomposition} equals $\begin{pmatrix}
        0_{n \times s} \\
        \alpha C \otimes \mathbf{1}_{2n+1}
    \end{pmatrix}
    \begin{pmatrix}
        \frac{1}{\alpha} D \otimes \mathbf{1}_{2n+1}^T & 0_{s \times n}
    \end{pmatrix}$ for any non-zero real $\alpha$, since $(C \otimes \mathbf{1}_{2n+1}) (D \otimes \mathbf{1}_{2n+1}^T) = T_k \otimes J_{(2n+1) \times (2n+1)}$.
    
    Consider the second term on the right-hand side of \eqref{eq_t_decomposition}. Let us introduce the matrices $\overline{A} \in \mathbb{R}^{(n+1) \times r}$ and $\overline{B} \in \mathbb{R}^{r \times (n+1)}$ as $\overline{A} = \begin{pmatrix} EA \\ 0_{1 \times r} \end{pmatrix}$ and $\overline{B} = \begin{pmatrix} 0_{r \times 1} & BE \end{pmatrix}$. In this notation, we have $\widetilde{A} = \diag(-\overline{A}, A)$, $\widetilde{B} = \diag(B, \overline{B})$. Moreover, 
    \begin{equation}
    \begin{split}
        \begin{pmatrix}
            A & 0_{n \times 2kr}\\
            0_{k(2n+1) \times r} & I_{k} \otimes \widetilde{A}
        \end{pmatrix} &= 
        \diag(\underbrace{A, -\overline{A}, A, -\overline{A}, \ldots, A}_{2k+1 \text{ blocks}}), \\
        \begin{pmatrix}
            I_k \otimes \widetilde{B} & 0_{2kr \times n} \\
            0_{r \times k(2n+1)} & B
        \end{pmatrix} &=
        \diag(\underbrace{B, \overline{B}, B, \overline{B}, \ldots, B}_{2k+1 \text{ blocks}}).
    \end{split}
    \end{equation}

    We claim $\overline{A} \; \overline{B} = U_{n+1}$. Indeed, we have $(EABE)_{ij} = (AB)_{n-i+1, n-j + 1} = (T_n)_{n-i+1, n-j + 1} = (T_n^T)_{ij}$. Together with zero paddings, we have $\overline{A} \; \overline{B} = \begin{pmatrix} 0_{n \times 1} & T_n^T \\ 0_{1 \times 1} & 0_{1 \times n} \end{pmatrix} = U_{n+1}$.
    Therefore, for our second term, we have 
    \begin{equation}
        \diag(\underbrace{T_n, -U_{n+1}, T_n, -U_{n+1}, \ldots, T_n}_{2k+1 \text{ blocks}}) = \diag(\underbrace{A, -\overline{A}, A, -\overline{A}, \ldots, A}_{2k+1 \text{ blocks}}) \times \diag(\underbrace{B, \overline{B}, B, \overline{B}, \ldots, B}_{2k+1 \text{ blocks}}).
    \end{equation}

    Finally, use the stacking trick: if we have any two matrix factorizations $M_1 = L_1 R_1$ and $M_2 = L_2 R_2$, then $M_1 + M_2 = \begin{pmatrix}
        L_1 & L_2
    \end{pmatrix}
    \begin{pmatrix}
        R_1 \\
        R_2
    \end{pmatrix}$.
    Thus, with $L$ obtained by horizontally stacking $\diag(\underbrace{A, -\overline{A}, \ldots, A}_{2k+1 \text{ blocks}})$ and $\begin{pmatrix}
        0_{n \times s} \\
        \alpha C \otimes \mathbf{1}_{2n+1}
    \end{pmatrix}$, and $R$ obtained by vertically stacking $\diag(\underbrace{B, \overline{B}, \ldots, B}_{2k+1 \text{ blocks}})$ and $\begin{pmatrix}
        \frac{1}{\alpha} D \otimes \mathbf{1}_{2n+1}^T & 0_{s \times n}
    \end{pmatrix}$, we get a factorization of $T_{2kn+n+k}$, concluding the proof.
\end{proof}

\paragraph{Proof of Theorem~\ref{th_fundamental_inequalities}.}
Let us analyze the quality of the factorization produced by the new $\lozenge_\alpha$ factorization. Let $AB = T_n$ and $CD = T_k$, and let $(L, R) := (A, B) \lozenge_\alpha (C, D)$. We have shown that $LR = T_{2kn + n + k}$. Clearly $\|A\|_{2 \to \infty} = \|EA\|_{2 \to \infty}$, $\|A\|_{F} = \|EA\|_{F}$, and $\|B\|_{1 \to 1} = \|BE\|_{1 \to 1}$. We observe the following:
\begin{equation}
    \begin{split}
        \|L\|_{2 \to \infty} &= \sqrt{\|A\|_{2 \to \infty}^2 + \alpha^2 \|C\|_{2 \to \infty}^2}, \\
        \frac{1}{\sqrt{2kn+n+k}} \|L\|_F &= \sqrt{\frac{(2k+1) \|A\|_F^2 + (2n+1) \alpha^2 \|C\|_F^2}{2kn+n+k}} < \sqrt{\frac{\|A\|_F^2}{n} + \frac{\alpha^2 \|C\|_F^2}{k}}, \\
        \|R\|_{1 \to 1} &= \|B\|_{1 \to 1} + \frac{1}{\alpha} \|D\|_{1 \to 1}.
    \end{split}
\end{equation}
Then we can compute the $LR$ factorization costs:
\begin{equation}
    \label{eq_quality_with_alpha}
    \begin{split}
        \|L\|_{2 \to \infty} \|R\|_{1 \to 1} &= \sqrt{\|A\|_{2 \to \infty}^2 + \alpha^2 \|C\|_{2 \to \infty}^2} \left( \|B\|_{1 \to 1} + \frac{1}{\alpha} \|D\|_{1 \to 1} \right), \\
        \frac{1}{\sqrt{2kn + n + k}} \|L\|_F \|R\|_{1 \to 1} &< \sqrt{\frac{\|A\|_F^2}{n} + \frac{\alpha^2 \|C\|_F^2}{k}} \left( \|B\|_{1 \to 1} + \frac{1}{\alpha} \|D\|_{1 \to 1} \right).
    \end{split}
\end{equation}
We are free to choose the parameter $\alpha$ as we wish. In particular, we can choose values that minimize the right-hand sides of the bounds in \eqref{eq_quality_with_alpha}. Setting the derivative to zero and carrying out a straightforward calculation yields the optimal values
\begin{equation}
    \label{eq_optimal_alpha}
    \begin{split}
        \alpha &= \left( \frac{\|A\|_{2 \to \infty}^2 \|D\|_{1 \to 1}}{\|C\|_{2 \to \infty}^2 \|B\|_{1 \to 1}} \right)^{1/3} \text{ for the $c_2$-cost, and} \\ 
        \alpha &= \left( \frac{\frac{1}{n}\|A\|_{F}^2 \|D\|_{1 \to 1}}{\frac{1}{k}\|C\|_{F}^2 \|B\|_{1 \to 1}} \right)^{1/3} \text{ for the $c_F$-cost}.
    \end{split}
\end{equation}
Plugging in these values into inequalities \eqref{eq_quality_with_alpha}, we get
\begin{equation}
    \label{eq_lozenge_cost_bounds}
    \begin{split}
        \|L\|_{2 \to \infty} \|R\|_{1 \to 1} &\leq \left( \left( \|A\|_{2 \to \infty} \|B\|_{1 \to 1} \right)^{2/3} + \left( \|C\|_{2 \to \infty} \|D\|_{1 \to 1} \right)^{2/3} \right)^{3/2},\\
        \frac{1}{\sqrt{2kn+n+k}} \|L\|_F \|R\|_{1 \to 1} &\leq \left( \left( \frac{1}{\sqrt{n}}\|A\|_{F} \|B\|_{1 \to 1} \right)^{2/3} + \left( \frac{1}{\sqrt{k}}\|C\|_{F} \|D\|_{1 \to 1} \right)^{2/3} \right)^{3/2}.
    \end{split}
\end{equation}
From now on, we will write $(A, B) \lozenge(C, D)$ without $\alpha$ in the subscript, meaning that $\alpha$ is chosen optimally for either $c_2$-cost or $c_F$-cost, when the cost is clear from the context. 

Theorem~\ref{th_fundamental_inequalities} is almost immediately implied. Suppose $AB$ and $CD$ are factorizations of $T_n$ and $T_k$ respectively that are arbitrarily close to the corresponding infimums, so
$(\|A\|_{2 \to \infty} \|B\|_{1 \to 1})^{2/3} \leq c_2(T_n)^{2/3} + \varepsilon$, and $(\|C\|_{2 \to \infty} \|D\|_{1 \to 1})^{2/3} \leq c_2(T_k)^{2/3} + \varepsilon$ for a positive $\varepsilon$. Then, for $(L, R) = (A, B) \lozenge (C, D)$, we have $(\|L\|_{2 \to \infty} \|R\|_{1 \to 1})^{2/3} \leq c_2(T_n)^{2/3} + c_2(T_k)^{2/3} + 2 \varepsilon$. Taking $\varepsilon$ arbitrarily small, we get $c_2(T_{2kn + n + k})^{2/3} \leq c_2(T_n)^{2/3} + c_2(T_k)^{2/3}$, as claimed in Theorem~\ref{th_fundamental_inequalities}. The $c_F$ cost is similar.

Now we have a tool to construct factorizations with good asymptotic behavior.
\begin{construction}
    \label{construction}
    Given a factorization $AB = T_k$ for fixed $k$, we construct a factorization $LR = T_n$ for arbitrary $n$ as follows:
    Let $(L',R') = \underbrace{(A,B) \lozenge \ldots \lozenge (A, B)}_{p \text{ times}}$, where the operations are performed left to right, and $p$ is the smallest number that yields the size of $L'R'$ greater than or equal to $n$. Remove rows of $L'$ after the $n$-th row to get $L$, and remove columns of $R'$ after the $n$-th column to get $R$.
\end{construction}

\begin{theorem}
    \label{th_costs_for_nonround_n}
    If there exists a factorization $AB = T_k$ for fixed $k$, such that 
    \begin{equation}
    \begin{split}
        \|A\|_{2 \to \infty} \|B\|_{1 \to 1} = C_2 (\log_2 k)^{3/2}, \\
        \frac{1}{\sqrt{k}}\|A\|_{F} \|B\|_{1 \to 1} = C_F (\log_2 k)^{3/2},
    \end{split}
    \end{equation}
    then for any $n \in \mathbb{N}$, Construction~\ref{construction} yields factorizations $LR = T_n$ such that
    \begin{equation}
        \begin{split}
        \|L\|_{2 \to \infty} \|R\|_{1 \to 1} &\leq C_2 \left( \frac{\log_2 k}{\log_2 (2k+1)} \right)^{3/2} (\log_2 n)^{3/2}(1+o(1)),\\
        \frac{1}{\sqrt{n}}\|L\|_{F} \|R\|_{1 \to 1} &\leq C_F \left( \frac{\log_2 k}{\log_2 (2k+1)} \right)^{3/2} (\log_2 n)^{3/2}(1+o(1)).
    \end{split}
    \end{equation}
\end{theorem}
In other words, if we have a good-quality factorization of $T_k$ in some fixed dimension $k$, then we can generate factorizations of $T_n$ in arbitrarily high dimensions, with a strictly better constant in front of $(\log_2 n)^{3/2}$ in the asymptotics of the costs. For large $k$, the factor $\left( \frac{\log_2 k}{\log_2 (2k+1)} \right)^{3/2}$ is approximately $\left( 1 - \frac{3/2}{\log_2 k} \right)$. 
In the next paragraph, we give intuition for Theorem~\ref{th_costs_for_nonround_n}  and defer the full proof to Appendix~\ref{appendix_main_proof}.

Consider the case when the matrices $(L', R') = \underbrace{(A,B) \lozenge \ldots \lozenge (A, B)}_{p \text{ times}}$ have size of $L'R'$ equal to $n$. In this case, $L = L'$ and $R = R'$. Notice that $n = \Theta\left( (2k+1)^p \right)$, implying $p = \frac{\log_2 n}{\log_2 (2k+1)} (1+o(1))$. Then we can apply equation \eqref{eq_lozenge_cost_bounds} $p-1$ times and get
    \begin{equation}
        \begin{split}
            (\|L'\|_{2 \to \infty} \|R'\|_{1 \to 1})^{2/3} \leq \left( C_2 (\log_2 k)^{3/2} \right)^{2/3} \cdot p = C_2^{2/3} \frac{\log_2 k}{\log_2 (2k+1)} (\log_2 n) (1+o(1)),
        \end{split}
    \end{equation}
    and similarly
    \begin{equation}
        \left( \frac{1}{\sqrt{N}} \|L'\|_{F} \|R'\|_{1 \to 1} \right)^{2/3} \leq C_F^{2/3} \frac{\log_2 k}{\log_2 (2k+1)} (\log_2 n) (1+o(1)).
    \end{equation}
    These give us the claim of the theorem. What is left to analyze is the case where $n$ is not exactly the size of $L'R'$. One would expect the same asymptotic bounds for these ``non-round'' $n$. There, the $c_2$ cost follows easily, but the $c_F$ needs a more careful argument. We bound the parameters $\alpha$ and exploit the recursive structure of $L'$. A complete analysis is given in Appendix~\ref{appendix_main_proof}.

\paragraph{Remark on the construction.}
The idea of constructing high-dimensional factorizations from lower-dimensional building blocks has appeared previously in the literature. Dvijotham, McMahan, Pillutla, Steinke, and Thakurta~\cite{dvijotham2024efficient} introduce the ``Recursive Matrix Factorization'' for the problem of memory-efficient approximate DP counting. The construction in~\cite{dvijotham2024efficient} takes factorizations of $T_n$ and $T_k$ and produces a factorization of $T_{kn}$. Thus, the logarithms of the dimensions are added: $\log n + \log k = \log(kn)$. In the resulting factorization, the quantities $\|.\|_{2 \to \infty}^2$ and $\|.\|_{1 \to 1}$ would also be added. In this way, one could start from a single good factorization and lift it to arbitrarily high dimensions while \emph{preserving} the constant in the asymptotics. We improve upon this approach by constructing a factorization of $T_{2kn+n+k}$, obtaining $\log n + \log k + \log 2 \approx \log(2kn+n+k)$. This additional $\log 2$ term allows us to achieve a better constant in our bounds.

\begin{table}[t]
\caption{We numerically optimize $c_2(T_k)$ and $c_F(T_k)$, obtaining constants $C_2$ and $C_F$; see the formulation of Theorem~\ref{th_costs_for_nonround_n}. We further present the resulting constants for MaxSE and MeanSE, adjusted for finite $k$.
}
\centering
\begin{tabular}{c|c|c|c|c|c|}
     $k \times s$ &  $2C^2_2$ & $2C^2_F$ & $(\frac{\log_2(k)}{\log_2(2k + 1)})^3$ & $2C^2_2(\frac{\log_2(k)}{\log_2(2k + 1)})^3$ &$2C^2_F(\frac{\log_2(k)}{\log_2(2k + 1)})^3$\\
     $1024 \times 1024$ & $0.1032$ & $0.0944$ & $0.7512$ & $0.0775$ & $0.0709$\\
     $512 \times 512$ & $0.1094$ & $0.1009$ & $0.7287$ & $0.0797$ & $0.0735$\\
     $256 \times 256$ & $0.1191$ & $0.1099$ & $0.7017$ & $0.0836$ & $0.0771$\\
     $128 \times 128$ & $0.1322$ & $0.1207$ & $0.6685$ & $0.0884$ & $0.0807$\\
     $64 \times 64$ & $0.1527$ & $0.1370$ & $0.6267$ & $0.0957$ & $0.0859$\\
     $32 \times 32$ & $0.1938$ & $0.1609$ & $0.5723$ & $0.1109$ & $0.0921$\\
     $16 \times 16$ & $0.2455$ & $0.2093$ & $0.4986$ & $0.1224$ & $0.1044$\\
     $8 \times 8$ & $0.3831$ & $0.3288$ & $0.3954$ & $0.1515$ & $0.1300$\\
     $4 \times 4$ & $0.8293$ & $0.6259$ & $0.2512$ & $0.2083$ & $0.1572$\\
\end{tabular}
\label{tab:numerical_optimization}
\end{table}

\subsection{Explicit factorizations and numerical results}

We perform numerical optimization to minimize the costs $c_2(T_k)$ and $c_F(T_k)$ over finite-dimensional matrices $L \in \mathbb{R}^{k \times s}$ and $R \in \mathbb{R}^{s \times k}$. 
In particular, we run $5\times 10^{4}$ iterations of the Adam optimizer~\cite{kingma2014adam}, followed by $300$ steps of the L-BFGS optimizer~\cite{liu1989limited}.
In Table~\ref{tab:numerical_optimization}, we report the constants for MaxSE and MeanSE. Specifically, we compute $C_2$ and $C_F$ from Theorem~\ref{th_costs_for_nonround_n} and report $2C_2^2$ and $2C_F^2$, which correspond to the constants in Table~\ref{tab:tree-methods}. Theorem~\ref{th_costs_for_nonround_n} shows that, by recursively lifting the dimension, the cost can be further improved by a factor of $\left(\frac{\log_2(k)}{\log_2(2k+1)}\right)^{3/2}$. This yields a larger improvement for smaller values of $k$. Taking the square of this improvement factor, we report the final MaxSE and MeanSE constants in the last two columns. The best performance was achieved by matrices of size $1024\times 1024$.

The matrices are, however, computed in finite precision and, by themselves, do not constitute a rigorous upper bound. In fact, strictly speaking, we do not obtain an exact factorization, since the product $LR$ is only approximately equal to $T_k$. To make the resulting bound rigorous, we provide the following reasoning.

\begin{proposition}
\label{prop:floating_point_bound}
    There exists a factorization $LR = T_k$ for $k = 1024$ with $\|L\|_{2 \to \infty} \le 1.8730$, $\|R\|_{1\to 1} \le 3.8417$, and therefore 
    \begin{equation}
        2C_2^2 = \frac{2\|L\|_{2 \to \infty}^2\|R\|_{1\to 1}^2}{\log_2(k)^3} \le 0.1036 \quad \text{and}\quad 2C_2^2\left(\frac{\log_2(k)}{\log_2(2k + 1)}\right)^3 \le 0.0778.
    \end{equation}
\end{proposition}

\begin{proof}

In the GitHub repository,\footnote{\url{https://github.com/npkalinin/Pure_DP_Continual_Counting_with_MF}} we provide an optimized float32 matrix $R$ of size $1024\times 1024$ for MaxSE. We compute the corresponding matrix $\tilde L$ numerically as $T_kR^{-1}$, using float64 arithmetic for better accuracy. We then convert the matrices $R$ and $\tilde L$ into rational matrices with denominator $2^{32}$, obtaining an approximate rational factorization. After this conversion, we no longer have that $\tilde L R = T_k$. Instead, we compute the residual $\tilde\Delta := \tilde L R - T_k$ exactly using rational arithmetic.

For the true matrix $L := T_kR^{-1} = \tilde{L} + \Delta L$, we can bound the $\|\cdot\|_{2\to \infty}$ norm as
    \begin{equation}
        \|L\|_{2\to \infty}
        \le
        \|\tilde{L}\|_{2\to \infty}
        +
        \|\Delta L\|_{2\to \infty}
        \le
        \|\tilde{L}\|_{2\to \infty}
        +
        \|\Delta L\|_{F}.
    \end{equation}

    The first norm can be bounded precisely. Namely, we use symbolic arithmetic to compute $\|\tilde{L}\|_{2\to \infty}^2$ as a fraction. We then compute a rational approximation to the square root with a denominator $2^{100}$, introducing an error of at most $2^{-100}$, and verify that its square is larger than the squared maximum $\ell_2$ row norm. In the remainder of the proof, we omit the details of the square-root computation; it should be understood in this sense.

    We further compute symbolically the expression $\tilde{\Delta} = \tilde{L}R - T_k$, from which we can express $R^{-1}$ as follows:\footnote{We found the matrices too large to compute $R^{-1}$ symbolically, and therefore rely on bounding the errors introduced by floating-point arithmetic.}
    \begin{equation}
        R^{-1}
        =
        (T_k + \tilde{\Delta})^{-1}\tilde{L}
        =
        \bigl(T_k(I_k + T_k^{-1}\tilde{\Delta})\bigr)^{-1}\tilde{L}
        =
        (I_k + T_k^{-1}\tilde{\Delta})^{-1}T_k^{-1}\tilde{L}.
    \end{equation}

    This allows us to express $\Delta L$:
    \begin{equation}
        \begin{aligned}
            \Delta L
            &=
            L - \tilde{L}
            =
            T_kR^{-1} - \tilde{L} =
            T_k(I_k + T_k^{-1}\tilde{\Delta})^{-1}T_k^{-1}\tilde{L}
            -
            \tilde{L} \\
            &=
            T_k
            \bigl[(I_k + T_k^{-1}\tilde{\Delta})^{-1} - I_k\bigr]
            T_k^{-1}\tilde{L} =
            T_k
            \left[
            \sum_{j = 1}^{\infty}
            (-1)^j (T_k^{-1}\tilde{\Delta})^{j}
            \right]
            T_k^{-1}\tilde{L}.
        \end{aligned}
    \end{equation}
    The last equation also requires verifying that $\|T_k^{-1}\tilde{\Delta}\|_F < 1$, which we do using exact rational arithmetic. We upper-bound the Frobenius norm of $\Delta L$ as
    \begin{equation}
        \|\Delta L\|_F
        \le
        \|T_k\|_F
        \|T_k^{-1}\tilde{L}\|_F
        \sum_{j = 1}^{\infty}
        \|T_k^{-1}\tilde{\Delta}\|_F^j
        =
        \frac{
        \|T_k\|_F
        \|T_k^{-1}\tilde{L}\|_F
        \|T_k^{-1}\tilde{\Delta}\|_F
        }{
        1 - \|T_k^{-1}\tilde{\Delta}\|_F
        }.
    \end{equation}

    All the quantities in this expression can be upper-bounded using precise symbolic computation, since the matrices $T_k$ and $T_k^{-1}$ are integer matrices. Carrying out this computation, we obtain the upper bound $\|\Delta L\|_F \le 0.0039$. 

    Combining this with the bound $\|\tilde{L}\|_{2\to \infty} \le 1.8691$ gives
    \begin{equation}
        \|L\|_{2 \to \infty} \le 1.8691 +0.0039 \le 1.8730.
    \end{equation}
    Computing the $\|R\|_{1\to 1}$ symbolically proves the bound in the statement, the final bounds for $2C_2^2$ and $2C_2^2\left(\frac{\log_2(k)}{\log_2(2k + 1)}\right)^3 $ follow directly by substituting the previous bounds and $k=1024$, concluding the proof.
\end{proof}

For MeanSE, we provide the corresponding floating-point matrix $R$ at the same repository.\footnote{\url{https://github.com/npkalinin/Pure_DP_Continual_Counting_with_MF}} The adjustments for floating-point errors are the same as in Proposition~\ref{prop:floating_point_bound}, with $\|L\|_F$ computed instead of $\|L\|_{2\to \infty}$. Repeating the same steps, we obtain
\begin{equation}
    \begin{aligned}
        &\|R\|_{1\to 1} \le 14.3098,
        \qquad
        \|\tilde{L}\|_{F} \le 15.3624,
        \qquad
        \|\Delta L\|_F \le 0.0037,
        \qquad
        \|L\|_F \le 15.3661, \\
        &2C_F^2
        =
        \frac{2\|L\|_F^2\|R\|_{1\to 1}^2}{k\log_2(k)^3}
        \le 0.0945,
        \qquad
        2C_F^2
        \left(\frac{\log_2(k)}{\log_2(2k + 1)}\right)^3
        \le 0.0710.
    \end{aligned}
\end{equation}
This gives the constants in Theorem~\ref{thm:main_theorem}. At the same link, we provide code for verifying the stated bounds using the SymPy library~\cite{sympy2017} for symbolic computation.

\section{Memory efficient noise correlation}
In this section, we prove the space and time complexity of the matrix factorization mechanism, with the corresponding factorization given by Construction~\ref{construction}. Since the resulting matrices are sparse and recursively constructed, the factorization naturally yields a recursively described algorithm that uses $O(\log n)$ memory and runs in $O(n \log n)$ time, matching the complexity of previous tree-based
algorithms.

Recall that the output of the matrix mechanism is
\begin{equation}
    \mathcal M_{L,R}(x)
    :=
    L(Rx + z) = T_n x + Lz,
\end{equation}
where $z$ is the noise vector with coordinates sampled independently from the Laplace distribution with scale $\|R\|_{1 \to 1}/\epsilon$. 

In this section, we fix a base factorization $AB = T_k$ for Construction~\ref{construction}, with $A \in \mathbb{R}^{k \times r}$, $B \in \mathbb{R}^{r \times k}$. If $p$ is the number of terms in the product $(A,B) \lozenge \ldots \lozenge (A,B) = (L', R')$, then the matrix $L'$ has size $\frac{(2k+1)^p - 1}{2} =: N_p$ times $\left( r + \frac{k+r}{2k} \right)(2k+1)^{p-1} - \frac{k+r}{2k} =: M_p$, which can be verified inductively. We denote $L_1 = A$, $R_1 = B$, and $(L_{i+1}, R_{i+1}) = (L_i, R_i) \lozenge (A, B)$. In this notation, $L' = L_p$, $R' = R_p$.
The main ingredient is an algorithm that streams the correlated noise $Lz$. Notice that the stream $Lz$ is just a prefix of a stream $L'z$. The noise-streaming algorithm is presented as Algorithm~\ref{alg_noise}.

\begin{algorithm}[!h]
    \DontPrintSemicolon
    \SetNoFillComment
    \tcc{Returns a stream of $L'z$ values of length $N_p$}

    \If{p = 1}{
        sample $z \sim Laplace(0,b)^r$ \;
        output $Az$ \;
    } \Else {
        sample $w \sim Laplace(0,b)^r$ \;
        $g \gets \alpha_{p-1} A w \in \mathbb{R}^k$ \;
        output {\tt Noise}($p-1$, $(A,B)$, $b$) \;
        \For{$i = 1, \ldots, k$}{
            output $g_i - ${\tt NoiseReversed}($p-1$, $(A,B)$, $b$) \;
            output $g_i$ \hspace{2cm}\tcc{the extra zero row in $\widetilde{A}$}
            output $g_i + ${\tt Noise}($p-1$, $(A,B)$, $b$) \;
        }
    }
    
    \caption{{\tt Noise}($p \in \mathbb{N}$, base factorization $(A,B)$, noise magnitude $b$)}
    \label{alg_noise}
\end{algorithm}
\begin{algorithm}[!h]
    \DontPrintSemicolon
    \SetNoFillComment
    \tcc{Returns a reversed stream of $L'z$ values of length $N_p$}

    \If{p = 1}{
        sample $z \sim Laplace(0,b)^r$ \;
        output reversed $Az$ \;
    } \Else {
        sample $w \sim Laplace(0,b)^r$ \;
        $g \gets \alpha_{p-1} A w \in \mathbb{R}^k$ \;
        \For{$i = k, k-1, \ldots, 1$}{
            output $g_i + ${\tt NoiseReversed}($p-1$, $(A,B)$, $b$) \;
            output $g_i$ \hspace{2cm}\tcc{the extra zero row in $\widetilde{A}$}
            output $g_i - ${\tt Noise}($p-1$, $(A,B)$, $b$) \;
        }
        output {\tt NoiseReversed}($p-1$, $(A,B)$, $b$) \;
    }
    
    \caption{{\tt NoiseReversed}($p \in \mathbb{N}$, base factorization $(A,B)$, noise magnitude $b$)}
    \label{alg_noise_reversed}
\end{algorithm}
In Algorithm~\ref{alg_noise}, {\tt NoiseReversed} streams the same correlated noise, but in the reversed order, the procedure is implemented in Algorithm~\ref{alg_noise_reversed} in a similar recursive way. $\alpha_{p-1}$ is the optimal parameter in the operation $(\underbrace{(A,B) \lozenge \ldots \lozenge (A,B)}_{p-1 \text{ times}}) \lozenge_{\alpha_{p-1}} (A, B) = (L', R')$. 
We prove that {\tt Noise} indeed streams the values of $L'z$, by induction: when $p=1$ we output $Az$ by construction, and when $p > 1$ the output is $\begin{pmatrix}
            L_{p-1} & 0_{N_{p-1} \times k(2M_{p-1}+1)} & 0_{N_{p-1} \times s}\\
            0_{k(2N_{p-1}+1) \times M_{p-1}} & I_{k} \otimes \widetilde{L}_{p-1} & \alpha_{p-1} A \otimes \mathbf{1}_{2N_{p-1}+1}
        \end{pmatrix} z$, which agrees with our recursive Construction~\ref{construction}. Similarly, {\tt NoiseReversed} steams the values of $L'z$ in reverse. 
The final mechanism is now straightforward: compute the prefix sums and add the values from the {\tt Noise} stream.
\begin{algorithm}[!h]
    \DontPrintSemicolon
    \SetNoFillComment
    \tcc{Returns a stream of $\epsilon$-differentially private prefix sums of $x$}

    $p \gets$ minimal number such that $N_p \geq n$ \;
    $\Delta \gets \|R'\|_{1 \to 1}$ \;
    initialize stream $G = ${\tt Noise}($p$, $(A,B)$, $\Delta / \epsilon$) \;
    $S \gets 0$ \;
    \For{$t = 1, \ldots, n$}{
        receive $x_t$ \;
        $S \gets S + x_t$ \;
        $\eta = \text{next}(G)$ \;
        output $S + \eta$ \;
    }
    
    \caption{{\tt PureDPContinualCounting}(input stream $x$, $n$, $\epsilon$, base factorization $(A,B)$)}
    \label{alg_main}
\end{algorithm}

In the algorithms, we use the values $\alpha_{p-1}$, $N_p$, $\|R'\|_{1 \to 1}$. Notice that these values can be computed recursively in $O(p)$ arithmetic operations.

\begin{theorem}
    Algorithm~\ref{alg_main} uses $O\left(\frac{r}{\log k} \log n \right) = O_{A,B}(\log n)$ space and $O \left( n \left( \frac{\log n}{\log k} + r \right) \right) = O_{A,B}(n \log n)$ time.
\end{theorem}
\begin{proof}
    First, we have $p = \frac{\log(n)}{\log(2k+1)} (1+o(1))$. It is enough to analyze the space and time complexity of {\tt Noise} and {\tt NoiseReversed}, because keeping track of the prefix sum $S$ (lines 4 and 7 of Algorithm~\ref{alg_main}) takes constant space and linear time.

    We analyze {\tt Noise} and {\tt NoiseReversed} by induction on $p$. Let $S_p$ and $\mathcal{T}_p$ be the space and time complexities of ${\tt Noise}(p, (A,B), b)$, for given $A, B$. The complexities for {\tt NoiseReversed} are the same.
    
    For $p = 1$, we need to store $S_1 = O(r)$ numbers to remember $z$ and $Az$ in both procedures. The calculation of $Az$ takes $\mathcal{T}_1 = O(kr)$ arithmetic operations (multiplications and additions).

    For higher $p$, in both procedures we sample the next $r$ noise values for $z$, and at each time have at most one instance of ${\tt Noise}(p-1, (A,B), b)$ or ${\tt NoiseReversed}(p-1, (A,B), b)$. Thus, we have $S_p = S_{p-1} + r$ and $\mathcal{T}_p = (2k+1) \mathcal{T}_{p-1} + O(N_p)$, where $O(.)$ hides the same constant for all $p$. 

    Space: the recurrence is $S_p = S_{p-1} + r$, implying $S_p = O(pr) = O\left(\frac{r}{\log k} \log n\right) = O_{A,B}(\log n)$.

    Time: the recurrence is $\mathcal{T}_p = (2k+1) \mathcal{T}_{p-1} + O(N_p)$, where $\mathcal{T}_1 = O(kr)$, and $N_p = \Theta((2k+1)^p)$. If we expand the recurrence, we get 
    \begin{equation}
        \mathcal{T}_p = O\left((2k+1)^{p-1} \mathcal{T}_1 + (2k+1)^{p-2} N_2 + (2k+1)^{p-3} N_3 + \ldots + (2k+1) N_{p-1} + N_p \right).
    \end{equation}
    Since $(2k+1)^{p-a} N_{a} = \Theta((2k+1)^p)$ with $\Theta$ hiding constants in a fixed bounded range for all $a \in \{2, \ldots, p\}$, we get $\mathcal{T}_p = O\left( (2k+1)^p \cdot r + p (2k+1)^p \right) = O(N_p (p+r)) = O \left( n \left( \frac{\log n}{\log k} + r \right) \right)$. For fixed $r$ and $k$, it gets us $O_{A,B}(n \log n)$.
\end{proof}

\section{Lower bounds for 0/1 factorizations}

In this section, we discuss lower bounds on $c_2(T_n)$ and $c_F(T_n)$. All known constructions for these factorization costs show the behavior $\Theta \left( \log n \right)^{3/2}$; however, the lower bounds remain $\Omega(\log n)$. In the next two theorems, we show that if one considers the factorizations $LR = T_n$ such that the entries of $L$ and $R$ are in $\{0, 1\}$, then both factorization costs must be $\Omega(\log^{3/2} n)$. Notice that the cases of binary tree factorizations and $k$-ary tree factorizations fall into this category of $0/1$ factorizations. Indeed, computing a sum of a segment that corresponds to a node in the tree yields a row of the form $(0, \ldots, 0, 1, \ldots, 1, 0, \ldots, 0)$ in $R$. Adding several such sums to get a prefix sum means a row with zeros and ones in $L$.

It would be interesting to generalize these two theorems to obtain a true lower bound of $\Omega(\log^{3/2} n)$ on $c_2(T_n)$ and $c_F(T_n)$, without the restriction to 0/1 factorizations. However, such generalization remains elusive.

\begin{theorem}
    For all matrices $L$ and $R$ with entries in $\{0,1\}$ satisfying $T_n = LR$, we have 
    \begin{equation}
        \|L\|_{2 \to \infty} \|R\|_{1 \to 1} \geq 0.290 \log_2^{3/2}(n-1) = \Omega(\log^{3/2} n).
    \end{equation}
\end{theorem}
\begin{proof}
    If both matrices $L$ and $R$ consist of zeros and ones, we can associate them with collections of sets. Define two families of sets $A_i, B_i \subseteq [m]$, where $m$ is the inner dimension for the factorization $T_n = LR$. Let $A_i = \{k \; | \; R_{k,i+1} = 1\}$ and $B_i = \{k \; | \; L_{i,k} = 1\}$ for $1 \leq i \leq n-1$.

    Recall the ``skew Bollob\'as'' inequality \cite{skew_bollobas}: suppose we have two families of $N$ sets $A_1, \ldots, A_N$ and $B_1, \ldots, B_N$ satisfying $A_i \cap B_i = \varnothing$ for all $i$, and $A_i \cap B_j \neq \varnothing$ whenever $i < j$. If $|A_i| \leq a$ and $|B_i| \leq b$ for all $i$, we have the inequality
    \begin{equation}
        N \leq {a+b \choose a}.
    \end{equation}

    We wish to apply the skew Bollob\'as inequality to our families $A_i$, $B_i$. Since $\sum\limits_k L_{ik} R_{kj} = (T_n)_{ij}$, we have $|A_i \cap B_j| = (T_n)_{j, i+1} = 1_{[i < j]}$. Thus, we indeed have $A_i \cap B_i = \varnothing$ and $A_i \cap B_j \neq \varnothing$ for $i < j$. Denote $a := \max\limits_i |A_i|$ and $b := \max\limits_i |B_i|$. By the skew Bollob\'as inequality, we have 
    \begin{equation}
        \label{eq_bollobas_applied}
        n - 1 \leq {a+b \choose a}.
    \end{equation}
    Observe that if the row of $L$ corresponding to some $B_i$ witnesses $|B_i| = b$, then this row has $2$-norm at least $\sqrt{|B_i|} = \sqrt{b}$. Similarly, if the column of $R$ corresponding to some $A_i$ witnesses $|A_i| = a$, then this column has $1$-norm at least $|A_i| = a$. Therefore, we have $\|L\|_{2 \to \infty} \|R\|_{1 \to 1} \geq a \sqrt{b}$. It is now left to show that under the constraint \eqref{eq_bollobas_applied}, $a \sqrt{b} = \Omega(\log^{3/2} n)$.
    Let $r := \frac{b}{a}$. Taking logarithms in Eq.~\eqref{eq_bollobas_applied} and using the entropy bound on the binomial coefficient, we get
    \begin{equation}
        \log (n - 1) \leq \log {a+b \choose a} \leq a \log\left( \frac{a+b}{a} \right) + b \log\left( \frac{a+b}{b} \right) = a \left( \log(1 + r) + r \log(1 + r^{-1}) \right).
    \end{equation}
    Now for the quality of the factorization, we have
    \begin{equation}
        \|L\|_{2 \to \infty} \|R\|_{1 \to 1} \geq a \sqrt{b} = a^{3/2} \sqrt{r} \geq \frac{\sqrt{r}}{\left( \log(1 + r) + r \log(1 + r^{-1}) \right)^{3/2}} \cdot \log^{3/2}(n-1).
    \end{equation}
    The function $f(r) = \frac{\sqrt{r}}{\left( \log(1 + r) + r \log(1 + r^{-1}) \right)^{3/2}}$ is bounded away from zero for $r > 0$. It has a unique minimum at $r_0 \approx 5.236$ with the value $f(r_0) \approx 0.503002$. Thus, we have
    \begin{equation}
        \|L\|_{2 \to \infty} \|R\|_{1 \to 1} \geq 0.503 \log^{3/2}(n - 1) \geq 0.290 \log_2^{3/2}(n-1) = \Omega(\log^{3/2} n),
    \end{equation}
    as claimed.
\end{proof}

\begin{theorem}
    For all matrices $L$ and $R$ with entries in $\{0,1\}$ satisfying $T_n = LR$, we have 
    \begin{equation}
        \frac{1}{\sqrt{n}} \|L\|_F \|R\|_{1 \to 1} \geq \left( \frac{\log 2}{6} \right)^{3/2} \sqrt{1 - \frac{1}{n}} \log_2^{3/2}(n-1) = \Omega(\log^{3/2} n)
    \end{equation}
\end{theorem}
\begin{proof}
    As in the proof of the previous theorem, define the sets $A_i$ and $B_i$ for $1 \leq i \leq n-1$ as $A_i = \{k \; | \; R_{k,i+1} = 1\}$ and $B_i = \{k \; | \; L_{i,k} = 1\}$ for $1 \leq i \leq n-1$. As before, we have $A_i \cap B_i = \varnothing$ and $A_i \cap B_j \neq \varnothing$ for $i < j$, so $\{(A_i, B_i)\}$ is a skew Bollob\'as system.

    Let $\max\limits_i|A_i| =: a$ and $|B_i| =: b_i$. Notice that $a$ and $b_i$ are strictly positive. Denote the average the $b_i$'s by $b := \frac{1}{n-1} \sum\limits_{i = 1}^{n-1} b_i$. Observe that $\frac{1}{\sqrt{n}} \|L\|_F \|R\|_{1 \to 1} \geq \sqrt{\frac{n-1}{n}} \sqrt{b} \cdot a$. So, it is enough to bound the value of $a \sqrt{b}$ from below.

    We will use the weighted version of the skew Bollob\'as inequality from \cite{bollobas_type_inequalities}: if $\{(A_i, B_i) \; | \; i \in [N]\}$ is a skew Bollob\'as system, then
    \begin{equation}
        \sum_{i = 1}^N \left( (1 + |A_i| + |B_i|) {|A_i| + |B_i| \choose |A_i|} \right)^{-1} \leq 1.
    \end{equation}
    In our case, this implies 
    \begin{equation}
        \sum_{i = 1}^{n-1} \left( (1 + a + b_i) {a + b_i \choose a} \right)^{-1} \leq 1.
    \end{equation}

    \underline{Step 1.} Let us first show the upper bound 
    \begin{equation}
        \label{eq_inequality_binomial}
        (1 + a + b_i) {a + b_i \choose a} \leq e^{C a^{2/3} b_i^{1/3}}
    \end{equation}
    for some positive constant $C$ and any strictly positive integers $a$, $b_i$. Let $r = b_i / a$. The right-hand side of inequality \eqref{eq_inequality_binomial} becomes $e^{C a r^{1/3}}$. Take the logarithm of the left-hand side of inequality \eqref{eq_inequality_binomial}:
    \begin{equation}
        \log \left( (1 + a + b_i) {a + b_i \choose a} \right) = \log (1 + a + b_i) + \log {a + b_i \choose a}.
    \end{equation}
    To prove \eqref{eq_inequality_binomial}, it is enough to show $\log (1 + a + b_i) \leq C_1 a r^{1/3}$ and $\log {a + b_i \choose a} \leq C_2 a r^{1/3}$ for some constants $C_1$ and $C_2$.

    First, we will use a rough estimate $1 + a + b_i \leq 3 a^2 b_i = 3 a^3 r$, which holds if we have $a \geq 1$, $b_i \geq 1$. Then $\log (1 + a + b_i) \leq \log (3 a^3 r) = \log \left( (3^{1/3} a r^{1/3})^3 \right) = 3 \log \left(3^{1/3} a r^{1/3} \right) \leq 3^{4/3} a r^{1/3}$. Thus, taking $C_1 = 3^{4/3}$ is enough.

    Second, we will again use the entropy bound on $\log {a + b_i \choose a} $:
    \begin{equation}
        \log {a+b_i \choose a} \leq a \log\left( \frac{a+b_i}{a} \right) + b_i \log\left( \frac{a+b_i}{b_i} \right) = a \left( \log(1 + r) + r \log(1 + r^{-1}) \right).
    \end{equation}
    The function $r^{-1/3} \left( \log(1 + r) + r \log(1 + r^{-1}) \right)$ is bounded from above, with the maximum of approximately $1.5811$ at $r_0 \approx 5.236$. Thus, the inequality $a \left( \log(1 + r) + r \log(1 + r^{-1}) \right) \leq C_2 a r^{1/3}$ holds with $C_2 = 1.6$. This concludes the proof of inequality \eqref{eq_inequality_binomial}. Since $3^{4/3} + 1.6 \approx 5.93$, setting $C = 6$ is enough. This value is not fine-tuned.\footnote{If the reader is interested in a sharper value for $C$, numerical analysis shows that the function $\log \left( (1 + a + b_i) {a + b_i \choose a} \right) / (a^{2/3} b_i^{1/3})$ over positive integers $a$, $b_i$ has a global maximum at $a = 1$, $b_i = 13$ with value $2.27407$. Then we can take $C = 2.275$. This yields the final bound of $\frac{1}{\sqrt{n}} \|L\|_F \|R\|_{1 \to 1} \geq 0.168 \sqrt{1 - \frac{1}{n}} \log_2^{3/2}(n-1)$.}

    \underline{Step 2.} 
    With inequality \eqref{eq_inequality_binomial} at hand, we can return to analyzing the skew Bollob\'as bound:
    \begin{equation}
        1 \geq \sum_{i = 1}^{n-1} \left( (1 + a + b_i) {a + b_i \choose a} \right)^{-1} \geq \sum_{i = 1}^{n-1} e^{- C a^{2/3} b_i^{1/3}} \geq (n - 1) e^{- C a^{2/3} b^{1/3}}.
    \end{equation}
    The last inequality above is Jensen's inequality, which applies since the function $f(x) = e^{- C a^{2/3} x^{1/3}}$ is convex for $x > 0$. This gives us
    \begin{equation}
    \begin{split}
        a^{2/3} b^{1/3} \geq \frac{1}{C} \log(n-1), \\
        a \sqrt{b} \geq \left( \frac{\log(n-1)}{C} \right)^{3/2} = \Omega(\log^{3/2} n).
    \end{split}
    \end{equation}
    Thus, we have $\frac{1}{\sqrt{n}} \|L\|_F \|R\|_{1 \to 1} \geq \sqrt{1 - \frac{1}{n}} \frac{\log(n-1)^{3/2}}{C^{3/2}} = \left(\frac{\log 2}{6}\right)^{3/2}\sqrt{1 - \frac{1}{n}} \log_2^{3/2}(n-1) = \Omega(\log^{3/2} n)$.
\end{proof}

\section{Additional results on factorization costs}
\label{sec_costs_not_norms}

Our factorization costs $c_2$ and $c_F$ are defined as $c_2(M) = \inf \{\|L\|_{2 \to \infty} \|R\|_{1 \to 1} \; | \; M=LR\}$ and $c_F(M) = \inf \{\frac{1}{\sqrt{n}}\|L\|_{F} \|R\|_{1 \to 1} \; | \; M=LR\}$ for an $n \times n$ matrix $M$. One would have a reasonable hope that $c_2$ and $c_F$ could be factorization norms. For example, their well-behaved and well-studied counterparts, the $\gamma_2$ and $\gamma_F$ norms, look very similar:
\begin{equation}
\begin{aligned}
    \gamma_2(M) = \inf \{\|L\|_{2 \to \infty} \|R\|_{1 \to 2} \; | \; M=LR\},\\
    \gamma_F(M) = \inf \{\tfrac{1}{\sqrt{n}}\|L\|_{F} \|R\|_{1 \to 2} \; | \; M=LR\}.
\end{aligned}
\end{equation}
In this section, we prove that neither $c_2$ nor $c_F$ is a norm, because they fail the triangle inequality. Our counterexample is the matrix decomposition $2T_2 = \begin{pmatrix}
    2 & 0 \\ 2 & 2
\end{pmatrix} = \begin{pmatrix}
    1 & 1\\1 & 1
\end{pmatrix}+\begin{pmatrix}
    1 & -1 \\ 1 & 1
\end{pmatrix}$. We will show that 
\begin{equation}
\label{eq_counterexample}
\begin{split}
    c_2 \begin{pmatrix}2 & 0\\2 & 2\end{pmatrix} &> c_2 \begin{pmatrix}1 & 1\\1 & 1\end{pmatrix} + c_2 \begin{pmatrix}1 & -1\\1 & 1\end{pmatrix} \text{ and} \\
    c_F \begin{pmatrix}2 & 0\\2 & 2\end{pmatrix} &> c_F \begin{pmatrix}1 & 1\\1 & 1\end{pmatrix} + c_F \begin{pmatrix}1 & -1\\1 & 1\end{pmatrix},
\end{split}
\end{equation}
which disproves the triangle inequality for $c_2$ and $c_F$.
The two matrices on the right-hand side can be factorized as $\begin{pmatrix}1 & 1\\1 & 1\end{pmatrix} = \begin{pmatrix}1 \\ 1\end{pmatrix} \begin{pmatrix}1 & 1\end{pmatrix}$ and $\begin{pmatrix}1 & -1\\1 & 1\end{pmatrix} = \begin{pmatrix}1 & -1\\1 & 1\end{pmatrix} \begin{pmatrix}1 & 0\\0 & 1\end{pmatrix}$. This implies
\begin{equation}
    \begin{split}
        c_2 \begin{pmatrix}1 & 1\\1 & 1\end{pmatrix} + c_2 \begin{pmatrix}1 & -1\\1 & 1\end{pmatrix} \leq 1 + \sqrt{2}, \\
        c_F \begin{pmatrix}1 & 1\\1 & 1\end{pmatrix} + c_F \begin{pmatrix}1 & -1\\1 & 1\end{pmatrix} \leq 1 + \sqrt{2}.
    \end{split}
\end{equation}
Thus, it is enough to prove $c_2(T_2) > \frac{1 + \sqrt{2}}{2}$ and $c_F(T_2) > \frac{1 + \sqrt{2}}{2}$. The rest of this section studies $c_2(T_2)$ and $c_F(T_2)$ to prove these inequalities.

First, we bound the inner dimension. We have the infimums taken over factorizations $T_2 = LR$ with $L \in \mathbb{R}^{2 \times r}$, $R \in \mathbb{R}^{r \times 2}$. We show in the next lemma that without loss of generality we can fix $r = 2$.

The proofs in this section are technical and moved to Appendix~\ref{appendix_proofs_not_norms}.
\begin{lemma}
    \label{lem_t2_inner_dim_bound}
    For any two matrices $L \in \mathbb{R}^{2 \times r}$, $R \in \mathbb{R}^{r \times 2}$, there exist matrices $L' \in \mathbb{R}^{2 \times 2}$, $R' \in \mathbb{R}^{2 \times 2}$, such that $L'R' = LR$ and 
    \begin{equation}
        \begin{split}
            \|L'\|_{2 \to \infty} \|R'\|_{1 \to 1} &\leq \|L\|_{2 \to \infty} \|R\|_{1 \to 1}, \\
            \|L'\|_{F} \|R'\|_{1 \to 1} &\leq \|L\|_{F} \|R\|_{1 \to 1}.
        \end{split}
    \end{equation}
\end{lemma}
The main idea of the proof is to act on the rows of $L$ and columns of $R$ with a suitable rotation. This preserves the dot products of the rows of $L$ and the columns of $R$, so the factorization remains valid. The task is to choose the rotation such that it does not stretch the columns of $R$ in terms of the $\ell_1$ norm, and also maps them into the span of the first two basis vectors. The complete proof is in Appendix~\ref{appendix_proofs_not_norms}.

Now we can find $c_2(T_2)$ and $c_F(T_2)$, limiting ourselves to $2 \times 2$ factors.
\begin{lemma}
    \label{lem_c2_t2}
    $c_2(T_2) = \sqrt{2}$.
\end{lemma}
\begin{lemma}
    \label{lem_cf_t2}
    $c_F(T_2) = \sqrt{3/2}$.
\end{lemma}
For both lemmas above, we show that the trivial factorization $T_2 = T_2 I_2$ is optimal. A useful trick for the proofs is that in two dimensions, the rotation by $45^\circ$ turns $\ell_1$ norms into $\ell_\infty$ norms, scaled by $\sqrt{2}$, which are easier to analyze. Both proofs essentially boil down to case studies. The proofs can be found in Appendix~\ref{appendix_proofs_not_norms}.

Both $\sqrt{2}$ and $\sqrt{3/2}$ are greater than $\frac{1 + \sqrt{2}}{2}$, so inequalities~\eqref{eq_counterexample} hold.

\subsection{General upper bound for matrix norms}
We have shown that $c_2$ and $c_F$ are not norms. Philosophically, this is part of the reason why we could not get $c_2(T_n)$ and $c_F(T_n)$ down to $O(\log n)$. In this section we show that if $c_2$ and $c_F$ \emph{were} norms, then $O(\log n)$ would be the right answer. We make the following general observation.

\begin{theorem}
    Let $\gamma$ be any matrix norm with the property $\gamma(st^T) = 1$ for all $s, t \in \{-1, 1\}^n$. Then $\gamma(T_n) \leq \lceil \log_2 n \rceil + 1$.
\end{theorem}
\begin{proof}
    First, we show that for any block matrix $M$ such that every block is filled either with ones or zeros, and every row and column contains at most one all-ones block, we have $\gamma(M) \leq 1$. Specifically, let $J_1, \ldots, J_p, K_1, \ldots, K_p \subseteq [n]$ be such that $J_i \cap J_j = \varnothing$ and $K_i \cap K_j = \varnothing$ for $i \neq j$, and let $M \in \{0, 1\}^{n \times n}$ be a matrix with $M_{jk} = 1$ if and only if $(j,k) \in (J_1 \times K_1) \sqcup \ldots \sqcup (J_p \times K_p)$. Then we will prove that $\gamma(M) \leq 1$.
    
    We use the probabilistic method. 
    
    Let $s$ and $t$ be random elements of $\{-1, 1\}^n$, conditioned on the equalities $s_j = t_k$ for all pairs of indices $(j,k) \in (J_1 \times K_1) \sqcup \ldots \sqcup (J_p \times K_p)$. We claim that $\mathbb{E} [s t^T] = M$. Indeed, for any $(j,k) \in (J_1 \times K_1) \sqcup \ldots \sqcup (J_p \times K_p)$, we have $(s t^T)_{jk} = s_j t_k = 1$, since $s_j$ and $t_k$ are either both $1$ or both $-1$. For any other $(j, k)$, we have $\mathbb{E}[(s t^T)_{jk}] = \mathbb{E}[- (s t^T)_{jk}] = 0$ out of symmetry. Now that we have $\mathbb{E} [s t^T] = M$, we can use Jensen's inequality:
    \begin{equation}
        \gamma(M) = \gamma(\mathbb{E} [s t^T]) \leq \mathbb{E} [\gamma(s t^T)] = 1,
    \end{equation}
    which holds since every norm is a convex function. Thus, we indeed have $\gamma(M) \leq 1$.

    Finally, with a standard trick, $T_n$ can be represented as a sum of $\lceil \log_2 n \rceil + 1 =: r$ such block matrices $T_n = M_1 + \ldots + M_r$ (for example, see Section 4.2 in \cite{matousek_gamma_2} and an illustration there). Using the triangle inequality, we have
    \begin{equation}
        \gamma(T_n) \leq \gamma(M_1) + \ldots + \gamma(M_r) \leq r = \lceil \log_2 n \rceil + 1,
    \end{equation}
    as claimed.
\end{proof}

We do not claim much novelty in this statement. For example, it can be deduced from Proposition 1.1 in \cite{Kwapien1970}. However, we find our proof above to be more elegant.

\paragraph{AI Declaration.} The authors used AI tools to assist in drafting the subsection ``Our contribution''
as a summary of the other sections; in drafting proofs of Lemmas
\ref{lem_t2_inner_dim_bound}, \ref{lem_c2_t2}, and \ref{lem_cf_t2};
in drafting technical parts of the proof of Theorem~\ref{th_costs_for_nonround_n};
and in preparing the code accompanying Lemma~\ref{prop:floating_point_bound}.

The authors have reviewed, edited, proofread, and verified all AI-assisted text, proofs, and code,
and take full responsibility for the content of the manuscript and accompanying materials.

\paragraph{Acknowledgment.}
We thank Rasmus Pagh and Joel Andersson for their valuable feedback on an earlier version of the draft, as well as for referring us to some related work. We additionally thank Rasmus for referring us to the code for averaged k-ary trees with subtraction.

Nikita Kalinin is supported in part by the Austrian Science Fund (FWF) [10.55776/COE12].

\bibliographystyle{plain}
\bibliography{bibliography}

\appendix

\section{Proof of the upper bound (Theorem~\ref{th_costs_for_nonround_n})}
\label{appendix_main_proof}

\begin{proof}[Proof of Theorem~\ref{th_costs_for_nonround_n}]
    In the construction, we first get the matrices $(L', R') = \underbrace{(A,B) \lozenge \ldots \lozenge (A, B)}_{p \text{ times}}$ with the size of $L'R'$ at least $n$. Let $N$ be such that $L'R' = T_N$, $N \geq n$. We have $N = \frac{(2k+1)^p - 1}{2}$, which can be verified by induction on $p$.

    Since $p$ is chosen as the smallest integer such that $n \leq \frac{(2k+1)^p - 1}{2}$, it is positive. We have $\frac{(2k+1)^{p-1} - 1}{2} < n \leq \frac{(2k+1)^p - 1}{2} = N$, implying $\log_2 N = (1 + o(1)) \log_2 n$.

    First, let us analyze the factorization costs of $T_N = L'R'$. Since $N = \frac{(2k+1)^p - 1}{2}$, we have $p = \frac{\log_2(2N+1)}{\log_2(2k+1)}$. Using this and equation \eqref{eq_lozenge_cost_bounds}, for the $c_2$-cost we get
    \begin{equation}
        \begin{split}
            (\|L'\|_{2 \to \infty} \|R'\|_{1 \to 1})^{2/3} &\leq \left( C_2 (\log_2 k)^{3/2} \right)^{2/3} \cdot p = C_2^{2/3} \log_2 k \cdot \frac{\log_2(2N+1)}{\log_2(2k+1)}\\
            &= C_2^{2/3} \frac{\log_2 k}{\log_2 (2k+1)} (\log_2 N) (1+o(1)),
        \end{split}
    \end{equation}
    \begin{equation}
        \|L'\|_{2 \to \infty} \|R'\|_{1 \to 1} \leq C_2 \left( \frac{\log_2 k}{\log_2 (2k+1)} \right)^{3/2} (\log_2 N)^{3/2} (1+o(1)).
    \end{equation}
    For the $c_F$-cost, with the same calculation we have
    \begin{equation}
    \begin{split}
        \left(\frac{1}{\sqrt{N}} \|L'\|_{F} \|R'\|_{1 \to 1}\right)^{2/3} &\leq \left( C_F (\log_2 k)^{3/2} \right)^{2/3} \cdot p = C_F^{2/3} \log_2 k \cdot \frac{\log_2(2N+1)}{\log_2(2k+1)}\\
        &= C_F^{2/3} \frac{\log_2 k}{\log_2 (2k+1)} (\log_2 N) (1+o(1)),
    \end{split}
    \end{equation}
    \begin{equation}
        \frac{1}{\sqrt{N}} \|L'\|_{F} \|R'\|_{1 \to 1} \leq C_F \left( \frac{\log_2 k}{\log_2 (2k+1)} \right)^{3/2} (\log_2 N)^{3/2} (1+o(1)).
    \end{equation}

    The two formulas above essentially say that if $n$ happens to be ``round'' in the sense $n = \frac{(2k+1)^p - 1}{2}$ for some $p$, then the theorem holds. The remainder of the proof is a careful analysis of what happens for $n$ that falls in between ``round'' values. To do so, we obtain $L$ and $R$ by truncating $L'$ after its $n$th row and $R'$ after its $n$th column. We now analyze the costs of the resulting factorization $T_n = LR$.
    \underline{$c_2$-cost.}
    Observe that $\|L\|_{2 \to \infty} \leq \|L'\|_{2 \to \infty}$ and $\|R\|_{1 \to 1} \leq \|R'\|_{1 \to 1}$. Thus, for the $c_2$-cost we have
    \begin{equation}
        \begin{split}
        c_2(T_n) \leq \|L'\|_{2 \to \infty} \|R'\|_{1 \to 1} &\leq C_2 \left( \frac{\log_2 k}{\log_2 (2k+1)} \right)^{3/2} (\log_2 N)^{3/2} (1+o(1))\\
        &= C_2 \left( \frac{\log_2 k}{\log_2 (2k+1)} \right)^{3/2} (\log_2 n)^{3/2}(1+o(1)),
    \end{split}
    \end{equation}
    as claimed.

    \underline{$c_F$-cost.}
    The $c_F$-cost is a bit trickier. It suffices to establish the inequality 
    \begin{equation}
        \label{eq_trimmed_L_frobenius}
        \frac{1}{\sqrt{n}} \|L\|_F \leq \frac{1}{\sqrt{N}} \|L'\|_F (1 + o(1)),
    \end{equation}
    which is not automatically implied by the fact that $\|L\|_F \leq \|L'\|_F$. If we had this inequality, then we could conclude that
    \begin{equation}
        \frac{1}{\sqrt{n}} \|L\|_{F} \|R\|_{1 \to 1} \leq (1+o(1)) \frac{1}{\sqrt{N}} \|L'\|_{F} \|R'\|_{1 \to 1} \leq C_F \left( \frac{\log_2 k}{\log_2 (2k+1)} \right)^{3/2} (\log_2 n)^{3/2}(1+o(1)),    
    \end{equation}
    and the proof would be done. Thus, we are going to bound $\frac{1}{\sqrt{n}} \|L\|_F$ in terms of $\frac{1}{\sqrt{N}} \|L'\|_F$.

    Let us analyze the sequence of pairs of matrices $(L_1, R_1), \ldots, (L_p, R_p)$, where $(L_1, R_1) = (A, B)$, $(L_{i+1}, R_{i+1}) = (L_i, R_i) \lozenge (A, B)$. Thus, $(L_p, R_p) = (L', R')$. Notice that $L_i$ has $N_i := \frac{(2k+1)^i - 1}{2}$ rows, and $R_i$ has $N_i := \frac{(2k+1)^i - 1}{2}$ columns.

    Recall that in our notation, $(L_i, R_i) \lozenge (A, B) = (L_i, R_i) \lozenge_\alpha (A, B)$ with $\alpha$ chosen optimally. Let $\alpha_i$ be the optimal parameter for the operation $(L_i, R_i) \lozenge (A, B)$. 
    \begin{claim}
        For all $i$, we have $\alpha_i \leq 1$.
    \end{claim}
    \begin{proof}
        We will use induction on $i$.
        
        \underline{Hypothesis:} 
        \begin{equation}
                \alpha_i \leq 1, \qquad
                \|L_i\|_F^2 \leq i \frac{(2k+1)^{i} - 1}{2k} \|A\|_F^2, \qquad
                \|R_i\|_{1 \to 1} \geq i \|B\|_{1 \to 1}
        \end{equation}
        for all $i \geq 1$.

        \underline{Base case:} for $i = 1$, we have $\alpha_1 = \left( \frac{\frac{1}{k}\|A\|_{F}^2 \|B\|_{1 \to 1}}{\frac{1}{k}\|A\|_{F}^2 \|B\|_{1 \to 1}} \right)^{1/3} = 1$. We also have $\|L_1\|_F^2 = \|A\|_F^2 = \frac{(2k+1)^{1} - 1}{2k} \|A\|_F^2$ and $\|R_1\|_{1 \to 1} = 1 \cdot \|B\|_{1 \to 1}$.

        \underline{Induction step:} Suppose the hypothesis holds for $i$, and consider $\|L_{i+1}\|_F$, $\|R_{i+1}\|_{1 \to 1}$ and $\alpha_{i+1}$:
        \begin{equation}
        \begin{split}
            \|L_{i+1}\|_F^2 &= (2k + 1) \|L_{i}\|_F^2 + \alpha_i^2 (2N_i + 1) \|A\|_F^2 \\
            &\leq \frac{\|A\|_F^2}{2k} \left( i(2k+1)^{i+1} + 2k(2k+1)^{i} - i(2k+1) \right)\\
            &\leq (i+1) \frac{(2k+1)^{i+1} - 1}{2k} \|A\|_F^2,
        \end{split}
        \end{equation}
        \begin{equation}
            \|R_{i+1}\|_{1 \to 1} = \|R_{i}\|_{1 \to 1} + \frac{1}{\alpha_i} \|B\|_{1 \to 1} \geq (i+1) \|B\|_{1 \to 1},
        \end{equation}
        \begin{equation}
            \alpha_{i+1} = \left( \frac{\frac{2}{(2k+1)^{i+1} - 1}\|L_{i+1}\|_{F}^2 \|B\|_{1 \to 1}}{\frac{1}{k}\|A\|_{F}^2 \|R_{i+1}\|_{1 \to 1}} \right)^{1/3} \leq 1,
        \end{equation}
        which concludes the induction step and the proof of the claim.
    \end{proof}

    Consider $\frac{1}{n} \|L\|_F^2 = \frac{1}{n} \sum_{j = 1}^n \|L'_{j,:}\|_2^2$ -- this is the average squared 2-norm of the first $n$ rows of $L'$. We need to show that it is close to the average of all rows of $L'$.

    First, we explain the intuition behind the discussed bound, and later we proceed with the formal proof. The value $\frac{1}{n} \|L\|_F^2$ is the expected value of the 2-norm of a random row of $L$. For any $j \in [n]$, from the structure of $L'$, we have that $\|L_{j,:}\|_2^2$ is a sum of at most $p$ values $\|A_{i,:}\|_2^2$ scaled by $\alpha^2$. We know that alphas are at most 1. Intuitively, the index $i$ in $\|A_{i,:}\|_2^2$ is distributed almost uniformly for most of the terms. Then one might expect $\left| \frac{1}{n} \|L\|_F^2 - \frac{1}{N} \|L'\|_F^2 \right| = O(1)$.
    
    Let us continue with the formal proof. For an integer $q\geq 1$, we write $u^{(q)}_j := \|(L_q)_{j,:}\|_2^2$ and $\mu_q := \frac{1}{N_q}\sum_{j=1}^{N_q} u^{(q)}_j = \frac{1}{N_q}\|L_q\|_F^2$,
where $N_q=((2k+1)^q-1)/2$. We wish to bound the value $\left| \frac{1}{n}\|L\|_F^2-\frac{1}{N}\|L'\|_F^2 \right| = \left| \frac{1}{n}\sum_{j=1}^n u^{(p)}_j-\mu_p \right|$ uniformly over $n$. For this reason, we also define $D_q(t):=\sum_{j=1}^t (u^{(q)}_j-\mu_q)$ and $F_q:=\max\limits_{1\leq t\leq N_q}|D_q(t)|$. Then it is enough to bound $F_q / n$.

We now prove the bound $F_q=O(N_{q-1})$ by induction. Let
$a_\ell:=\|A_{\ell,:}\|_2^2$ and $M_A:=\max_{\ell\in[k]} a_\ell$.
We also write $\operatorname{rev}(u^{(q)})$ for the sequence $u^{(q)}$
in reverse order.

From the recursive definition of $L_{q+1}$, the sequence of squared row
norms of $L_{q+1}$ is
\begin{equation}
\begin{split}
u^{(q+1)}
=
\big(&u^{(q)},\;
\operatorname{rev}(u^{(q)})+\alpha_q^2 a_1,\;
\alpha_q^2 a_1,\;
u^{(q)}+\alpha_q^2 a_1,\; \ldots, \\
&\operatorname{rev}(u^{(q)})+\alpha_q^2 a_k,\;
\alpha_q^2 a_k,\;
u^{(q)}+\alpha_q^2 a_k
\big),
\end{split}
\label{eq_row_norm_sequence_recurrence}
\end{equation}
where adding a scalar to a sequence means adding it to every coordinate.
Since $\alpha_q\leq 1$, all shifts $\alpha_q^2 a_\ell$ are bounded by
$M_A$. Also, every coordinate of $u^{(q)}$ is at most $qM_A$, and hence
$\mu_q\leq qM_A$.

We show that the averages $\mu_q$ change by at most a constant
at each step. Indeed, using \eqref{eq_row_norm_sequence_recurrence},
\begin{equation}
    \mu_{q+1}
    =
    \frac{(2k+1) N_q\mu_q+\alpha_q^2(2N_q+1)\sum_{\ell=1}^k a_\ell}
         {N_{q+1}}.
\end{equation}
Since $N_{q+1}=(2k+1) N_q+k$, $\alpha_q\leq 1$, and $\mu_q\leq qM_A$,
we get $|\mu_{q+1}-\mu_q|\leq C_0$ for some constant $C_0$ depending
only on the fixed base factorization and on $k$.

We claim that there is a constant $K$ such that $F_q\leq K N_{q-1}$ for all $q \geq 2$.

\underline{Base case.}
For $q=2$, the quantity $F_2$ is finite and $N_1=k>0$. Therefore
$F_q\leq K N_{q-1}$ holds for $q=2$ after choosing
$K\geq F_2/N_1$.

\underline{Induction step.}
Assume that $F_q\leq K N_{q-1}$ holds for some $q\geq 2$.
We prove it for $q+1$.

Consider an arbitrary prefix of the sequence $u^{(q+1)}$. By equation
\eqref{eq_row_norm_sequence_recurrence}, this prefix consists of
$O_k(1)$ complete pieces and at most one incomplete piece. The incomplete
piece is either a prefix of $u^{(q)}+c$ or a prefix of
$\operatorname{rev}(u^{(q)})+c$, where $0\leq c\leq M_A$.

For a prefix of $u^{(q)}+c$, centered around $\mu_{q+1}$, the discrepancy
is at most $F_q + N_q(M_A+C_0)$.
The same bound holds for a prefix of $\operatorname{rev}(u^{(q)})+c$,
because a prefix of the reversed sequence is a suffix of $u^{(q)}$, and
suffix discrepancies are similarly bounded by $F_q$. Each complete shifted
copy of $u^{(q)}$ or $\operatorname{rev}(u^{(q)})$ contributes $O(N_q)$,
and the singleton pieces contribute $O_k(q)=O(N_q)$. Hence there is a
constant $C_1$, depending only on the fixed base factorization and on
$k$, such that
\begin{equation}
    F_{q+1}\leq F_q+C_1N_q.
    \label{eq_Bq_recursion}
\end{equation}

Using the induction hypothesis and the fact that $N_{q-1}/N_q<1/(2k+1)$,
we get
\begin{equation}
    F_{q+1}
    \leq K N_{q-1}+C_1N_q
    \leq \left(\frac{K}{(2k+1)}+C_1\right)N_q.
\end{equation}
Choosing $K\geq C_1/(1-1/(2k+1))$ makes the last expression at most
$K N_q$. This proves the induction step. Thus $F_q=O(N_{q-1})$.

We now apply this estimate with $q=p$. Recall that
$L=L_p[1:n,:]$, $L'=L_p$, $N=N_p$, and by minimality of $p$ we have
$N_{p-1}<n\leq N_p$. Therefore
\begin{equation}
    \left|
    \frac{1}{n}\|L\|_F^2-\frac{1}{N}\|L'\|_F^2
\right|
=
\left|
    \frac{1}{n}\sum_{j=1}^n u^{(p)}_j-\mu_p
\right| 
=
\frac{|D_p(n)|}{n}
\leq
\frac{F_p}{n}
\leq
O\!\left(\frac{N_{p-1}}{n}\right)
=
O(1).
\end{equation}
Thus $\frac{1}{n}\|L\|_F^2 = \frac{1}{N}\|L'\|_F^2+O(1)$.
We have $\frac{1}{N}\|L'\|_F^2=\mu_p\to\infty$ as $p\to\infty$, since $\mu_p$ grows linearly with $p$. This implies
\begin{equation}
    \frac{1}{\sqrt n}\|L\|_F
    \leq
    \frac{1}{\sqrt N}\|L'\|_F(1+o(1)).
\end{equation}
Together with $\|R\|_{1\to 1}\leq \|R'\|_{1\to 1}$, this proves the
desired $c_F$-cost bound.
\end{proof}

\section{Proofs of the additional results}
\label{appendix_proofs_not_norms}

\begin{proof}[Proof of Lemma~\ref{lem_t2_inner_dim_bound}]
\mbox{}\\
    Let $a, b \in \mathbb{R}^r$ be the rows of $L$ and $c, d \in \mathbb{R}^r$ be the columns of $R$. Let $e_1, \ldots, e_r$ be the standard basis in $\mathbb{R}^r$.
    
    It is enough to find a rotation $S \in \text{SO}(r)$ that satisfies the following: $Sc, Sd \in \text{span}(e_1, e_2)$, $\|Sc\|_1 \leq \|c\|_1$, $\|Sd\|_1 \leq \|d\|_1$. Indeed, if we have such a rotation, then we have $T_2 = LR = (LS^T)(SR)$. Moreover, $\|SR\|_{1 \to 1} \leq \|R\|_{1 \to 1}$, and both columns of $SR$ are supported on the first two coordinates. Rotations preserve 2-norms of vectors, so we have $\|LS^T\|_{2 \to \infty} = \|L\|_{2 \to \infty}$ and $\|LS^T\|_{F} = \|L\|_{F}$. Then we can construct $L' = (LS^T)[1:2, 1:2]$, $R' = (SR)[1:2, 1:2]$ without loss in the costs of the factorization. In the remainder of the proof, we show that such a rotation exists.

Consider the subspace \(W:=\operatorname{span}(c,d)\). The case $\dim W = 1$ is obvious, so we will assume that $W$ is two-dimensional. Our goal is to choose an orthonormal basis \(u,v\) of \(W\) such that
\begin{equation}
    \label{eq_basis_properties_dim_reduction}
    |\langle x,u\rangle|+|\langle x,v\rangle|\leq \|x\|_1, \quad x\in\{c,d\}.
\end{equation}

Indeed, once such \(u,v\) are found, we may extend them to a (positively oriented) orthonormal basis $u,v,u_3,\ldots,u_r$
of \(\mathbb{R}^r\), and define \(S\in \operatorname{SO}(r)\) by $Su=e_1$, $Sv=e_2$, $Su_k=e_k$ for $(3\leq k\leq r)$. Then \(Sc,Sd\in \operatorname{span}(e_1,e_2)\), and the desired
\(\ell_1\)-bounds follow from \eqref{eq_basis_properties_dim_reduction}.

So it remains to choose \(u,v\). We use angles between lines, always taken in
\([0,\pi/2]\). For a nonzero vector \(x\), let \(\beta_x\) be the largest number in $[0,\pi/4]$ such that
\begin{equation}
    \cos\beta_x+\sin\beta_x
 \leq \frac{\|x\|_1}{\|x\|_2}.
\end{equation}

Notice that the function $\cos\beta + \sin\beta$ is increasing on $[0, \pi/4]$. The number $\beta_x$ has the following meaning: if \(x\) makes angle at most
\(\beta_x\) with one of the two coordinate lines \(\mathbb{R}u\) or
\(\mathbb{R}v\), then
\begin{equation}
    |\langle x,u\rangle|+|\langle x,v\rangle|
 \leq \|x\|_2(\cos\beta_x+\sin\beta_x)
 \leq \|x\|_1.
\end{equation}

Thus we need to choose the two lines \(\mathbb{R}u,\mathbb{R}v\) such that
\(\mathbb{R}c\) is within angle \(\beta_c\) of one of them, and
\(\mathbb{R}d\) is within angle \(\beta_d\) of one of them.

We observe that if a non-zero vector $x$ satisfies $\beta_x<\pi/4$, then there exists a basis vector $e_j$ such that $\angle(\mathbb{R}x,\mathbb{R}e_j) \leq \pi/4$. Suppose
\(\beta_x<\pi/4\), and normalize \(x\) so that \(\|x\|_2=1\). Then $\|x\|_1 = \cos \beta_x + \sin \beta_x < \sqrt{2}$. Choose $e_j$ such that $|x_j|$ is maximal. Then $|x_j| = \|x\|_\infty$. We have
\begin{equation}
    1 = \|x\|_2^2 \leq \|x\|_\infty \|x\|_1 < |x_j| \sqrt{2}.
\end{equation}
Then we have $|x_j| > \sqrt{2}/2$, and the angle
\(\varphi\) between \(\mathbb{R}x\) and \(\mathbb{R}e_j\) is in
\([0,\pi/4]\), with \(|x_j|=\cos\varphi\). Moreover,
\begin{equation}
     \|x\|_1 - |x_j|=\sum_{i\neq j}|x_i|
 \geq \left(\sum_{i\neq j}|x_i|^2\right)^{1/2}
 =\sqrt{1-|x_j|^2}
 =\sin\varphi.
\end{equation}
Therefore, $\|x\|_1 \geq \cos\varphi+\sin\varphi$.
By the definition of \(\beta_x\), this implies \(\varphi\leq \beta_x\).
Thus, whenever \(\beta_x<\pi/4\), the line \(\mathbb{R}x\) is within angle
\(\beta_x\) of some standard coordinate line.

Let $\alpha:=\angle(\mathbb{R}c,\mathbb{R}d)\in[0,\pi/2]$ and $\Delta:=\min\{\alpha,\pi/2-\alpha\} \in [0, \pi/4]$.
We claim that $\Delta\leq \beta_c+\beta_d$.

If either \(\beta_c\) or \(\beta_d\) equals \(\pi/4\), this is immediate
because \(\Delta\leq \pi/4\). Otherwise, by the previous paragraph, there are
standard coordinate lines \(\mathbb{R}e_i\) and \(\mathbb{R}e_j\) such that $\angle(\mathbb{R}c,\mathbb{R}e_i)\leq \beta_c$ and $\angle(\mathbb{R}d,\mathbb{R}e_j)\leq \beta_d$.

If \(i=j\), then by the triangle inequality $\alpha\leq \beta_c+\beta_d$ and \(\Delta\leq \beta_c+\beta_d\). Otherwise, the coordinate lines \(\mathbb{R}e_i\) and
\(\mathbb{R}e_j\) are orthogonal, and $\frac{\pi}{2} \leq \alpha + \beta_c + \beta_d$. Then we have $\frac{\pi}{2}-\alpha\leq \beta_c+\beta_d$, implying \(\Delta\leq \beta_c+\beta_d\).

We can now choose the basis of \(W\). If \(\alpha\leq \pi/4\), then
\(\alpha\leq \beta_c+\beta_d\), so we choose a line \(L\subset W\) between
\(\mathbb{R}c\) and \(\mathbb{R}d\) such that $\angle(L,\mathbb{R}c)\leq \beta_c$ and $\angle(L,\mathbb{R}d)\leq \beta_d$.
If \(\alpha>\pi/4\), let \(D^\perp\subset W\) be the line orthogonal to
\(\mathbb{R}d\). Then we have $\angle(\mathbb{R}c,D^\perp)=\frac{\pi}{2}-\alpha
 \leq \beta_c+\beta_d$,
so we can choose a line \(L\subset W\) between \(\mathbb{R}c\) and \(D^\perp\)
such that $\angle(L,\mathbb{R}c)\leq \beta_c$ and $\angle(L,D^\perp)\leq \beta_d$.

Finally, take a unit vector \(u\) spanning \(L\), and let \(v\) be a unit vector
in \(W\) orthogonal to \(u\). In the first case, both \(c\) and \(d\) are close
to the line \(\mathbb{R}u\). In the second case, \(c\) is close to
\(\mathbb{R}u\), while \(d\) is close to \(\mathbb{R}v\). Hence, for
\(x\in\{c,d\}\), we have
\begin{equation}
    |\langle x,u\rangle|+|\langle x,v\rangle|\leq \|x\|_1,
\end{equation}
as desired.
\end{proof}

\begin{proof}[Proof of Lemma~\ref{lem_c2_t2}]
First, observe that $c_2(T_2) \leq \sqrt{2}$, as witnessed by the trivial factorization $T_2 = T_2 I_2$. The rest of the proof shows that for any factorization $T_2 = LR$ we have $\|L\|_{2 \to \infty} \|R\|_{1 \to 1} \geq \sqrt{2}$.

Consider any factorization \(T_2=LR\), where $L$ and $R$ are $2 \times 2$ matrices. Since \(T_2\) is invertible, both \(L\) and \(R\) are invertible. Let $u, v \in \mathbb{R}^2$ be the two rows of \(L\), and let $w = v - u$.
Then
\begin{equation}
    \|L\|_{2\to\infty}
=
\max\{\|u\|_2,\|v\|_2\}
=
\max\{\|u\|_2,\|u+w\|_2\}.
\end{equation}

Let \(u=(u_1,u_2)\), \(w=(w_1,w_2)\), and $\Delta=u_1w_2-u_2w_1$. With a calculation, we have
\begin{equation}
    R = L^{-1} T_2
=
\frac1{\Delta}
\begin{pmatrix}
w_2 & -u_2\\
-w_1 & u_1
\end{pmatrix}.
\end{equation}
Therefore
\begin{equation}
    \|R\|_{1\to 1}
=
\frac{\max\{\|u\|_1,\|w\|_1\}}{|\Delta|}.
\end{equation}

It remains to estimate \(|\Delta|\). We would like to turn $\ell_1$ norms into $\ell_\infty$ norms. In two dimensions, we can do this with a rotation by $\pi/4$ and rescaling by matrix $S=\begin{pmatrix}1&1\\1&-1\end{pmatrix}$.
For every \(x\in\mathbb{R}^2\), we have \(\|Sx\|_\infty = \|x\|_1\) and $\|S x\|_2=\sqrt2\,\|x\|_2$. Also,
\begin{equation}
    |\det(S u,S w)|=2|\det(u,w)|=2|\Delta|.
\end{equation}
We need the following claim.

\begin{claim}
    For all \(x,y\in\mathbb R^2\), we have
\begin{equation}
    |\det(x,y)|
\le
\max\{\|x\|_\infty,\|y\|_\infty\}
\max\{\|x\|_2,\|x+y\|_2\}.
\end{equation}
\end{claim}
\begin{proof}
By homogeneity, it is enough to assume $\max\{\|x\|_\infty,\|y\|_\infty\}\le 1$. Let \(x=(a,b)\) and \(y=(e,f)\). Then $|\det(x,y)|=|af-be|$.
Without loss of generality, we can change signs of coordinates and swap the coordinates. Therefore, we can assume $a,b\ge 0$ and $af-be\ge 0$.
Let \(D=af-be\). Consider the following cases.

\begin{enumerate}
\item If \(e\ge 0\), then $D=af-be\le af\le a\le \|x\|_2$.
\item If \(f\le 0\), then $D=af-be\le -be\le b\le \|x\|_2$.
\item The only remaining case is \(e<0<f\). Then, since \(a\le 1\) and \(-e\le 1\), $D=af+b(-e)\le f+b$.
But \(b+f\) is the second coordinate of \(x+y\), so $D\le b+f\le \|x+y\|_2$.
\end{enumerate}
Thus the claim follows.
\end{proof}
 
Applying the claim to \(x=S u\) and \(y=S w\), we get
\begin{equation}
    2|\Delta|
\le
\max\{\|S u\|_\infty,\|S w\|_\infty\}
\max\{\|S u\|_2,\|S(u+w)\|_2\}.
\end{equation}
Thus
\begin{equation}
    2|\Delta|
\le
\max\{\|u\|_1,\|w\|_1\}\,
\sqrt2\max\{\|u\|_2,\|u+w\|_2\}.
\end{equation}
Therefore we have
\begin{equation}
    \|L\|_{2\to\infty}\,\|R\|_{1\to 1}
=
\max\{\|u\|_2,\|u+w\|_2\}
\frac{\max\{\|u\|_1,\|w\|_1\}}{|\Delta|}
\ge \sqrt2,
\end{equation}
which concludes the proof.
\end{proof}

\begin{proof}[Proof of Lemma~\ref{lem_cf_t2}]
First, $c_F(T_2) \leq \sqrt{3/2}$ as witnessed by the trivial factorization $T_2 = T_2 I_2$. We will now prove that $c_F(T_2) \geq \sqrt{3/2}$.

We will use the same rotation plus rescaling trick as in the previous lemma to turn $\ell_1$ norms into $\ell_\infty$ norms. Let $S=\begin{pmatrix}1&1\\1&-1\end{pmatrix}$.
Then we have \(\|x\|_1=\|Sx\|_\infty\) for every \(x\in\mathbb{R}^2\),
\(S^{-1}=\frac{1}{2}S\), and \(SS^T=2I\).

Consider any factorization \(T_2=LR\), where $L$ and $R$ are $2 \times 2$ matrices. Let \(\widehat{L}=LS^{-1}\) and \(\widehat{R}=SR\). Then
\(T_2=\widehat{L}\widehat{R}\), \(L=\widehat{L}S\), and for the norms of $L$ and $R$ we have
\begin{equation}
    \|L\|_F^2
=
2\|\widehat{L}\|_F^2 \qquad \text{and} \qquad
\|R\|_{1\to 1}
=
\max_j \|R_{:, j}\|_1
=
\max_j \|\widehat{R}_{:, j}\|_\infty
=
\|\widehat{R}\|_{\infty},
\end{equation}
where $\|\widehat{R}\|_{\infty}$ is the maximum absolute value of an entry of $\widehat{R}$. Therefore
\begin{equation}
\frac{1}{\sqrt{2}}\|L\|_F\|R\|_{1\to 1}
=
\|\widehat{L}\|_F\|\widehat{R}\|_{\infty}.
\end{equation}

By homogeneity, we may assume \(\|\widehat{R}\|_{\infty}=1\). Let $\widehat{R}=\begin{pmatrix}a&b\\ c&d\end{pmatrix}$ and $\Delta=ad-bc$. Then
\begin{equation}
\widehat{L} = T_2 \widehat{R}^{-1}
=
\frac{1}{\Delta}
\begin{pmatrix}
d&-b\\
d-c&a-b
\end{pmatrix},
\end{equation}
\begin{equation}
\|\widehat{L}\|_F^2 = \|T_2 \widehat{R}^{-1}\|_F^2
=
\frac{d^2+b^2+(d-c)^2+(a-b)^2}{(ad-bc)^2}.
\end{equation}

It is enough to prove that, whenever \(|a|,|b|,|c|,|d|\le 1\),
\begin{equation}
\label{eq_goal_cf_t2}
d^2+b^2+(d-c)^2+(a-b)^2
\ge
\frac{3}{2}(ad-bc)^2.
\end{equation}

Let \(x=(a,c)\), \(y=(b,d)\). Then \(x,y\in [-1,1]^2\), and
the desired inequality becomes
\begin{equation}
\Psi(x,y):=
\|y\|_2^2+\|x-y\|_2^2-\frac{3}{2}\det(x,y)^2
\ge 0.
\end{equation}
$\Psi$ is a non-convex function defined on a 4-dimensional cube. To prove the above inequality, we are going to consider quite a few cases. First, let $(x, y)$ be the global minimum of $\Psi$ on $[-1,1]^2 \times [-1,1]^2$.
\begin{enumerate}
    \item Suppose $x$ is in the interior of $[-1,1]^2$.

Let \(J = \begin{pmatrix} 0 & 1 \\ -1 & 0 \end{pmatrix}\), so that \(\det(x,y)=x Jy\). Then
\begin{equation}
\nabla_x\Psi=2(x-y)-3\det(x,y)\,Jy=0.
\end{equation}
Thus \(x-y=\frac{3}{2}\det(x,y)\,Jy\). Taking the dot product with \(Jy\), and using
\(y Jy=0\), gives \(\det(x,y)=\frac{3}{2}\det(x,y)\|y\|_2^2\). If \(\det(x,y)=0\), then
\(\Psi(x,y)\ge 0\). Otherwise, \(\|y\|_2^2=2/3\), and then
\(\|x-y\|_2^2=\frac{3}{2}\det(x,y)^2\), so \(\Psi(x,y)=2/3>0\).

\item Suppose $y$ is in the interior of $[-1,1]^2$.

Fix \(x\), and suppose that \(y\) is an interior minimizer. Let
\(J' = \begin{pmatrix} 0 & -1 \\ 1 & 0 \end{pmatrix}\), so that \(\det(x,y)=y J'x\). Then
\begin{equation}
\nabla_y\Psi=4y-2x-3\det(x,y)\,J'x=0.
\end{equation}
Taking the dot product with \(J'x\) gives \(4\det(x,y)=3\det(x,y)\|x\|_2^2\). If \(\det(x,y)=0\), then
\(\Psi(x,y)\ge 0\). Otherwise, \(\|x\|_2^2=4/3\), and
\(y=\frac{1}{2}x+\frac{3}{4}\det(x,y)\,J'x\). Therefore
\(\|y\|_2^2=\|x-y\|_2^2=\frac{1}{3}+\frac{3}{4}\det(x,y)^2\), so
\(\Psi(x,y)=2/3>0\).

\item The remaining case is when both \(x\) and \(y\) are on the boundary of
\([-1,1]^2\). Since \(\Psi\) is invariant under the symmetries of the square, we may
assume that \(y=(t,1)\), where \(-1\le t\le 1\). The point \(x\) lies on one of
the four edges of \([-1,1]^2\). Consider cases.
\begin{enumerate}
    \item 
If \(x=(s,1)\), then
\begin{equation}
\Psi(x,y)
=
\frac{1}{2}\left(2-s^2+2st+t^2\right).
\end{equation}
For fixed \(t\), this is concave in \(s\), so its minimum on \([-1,1]\) occurs
at \(s=\pm 1\). The endpoint values are \(\frac{1}{2}(t+1)^2\) and
\(\frac{1}{2}(t-1)^2\), hence \(\Psi(x,y)\ge 0\).

\item If \(x=(s,-1)\), then
\begin{equation}
\Psi(x,y)
=
\frac{1}{2}\left(10-s^2-10st+t^2\right).
\end{equation}
Again this is concave in \(s\), so it is enough to check \(s=\pm 1\). The
endpoint values are \(\frac{1}{2}(t-1)(t-9)\) and
\(\frac{1}{2}(t+1)(t+9)\), both of which are nonnegative for \(-1\le t\le 1\).

\item If \(x=(1,s)\), we have
\begin{equation}
\Psi(x,y) = \Psi((1,s), (t,1)) = -\frac{3}{2}s^2t^2 + s^2 + 3st - 2s + 2t^2 - 2t + \frac{3}{2}.
\end{equation}
We need to verify that, as a function of $s$ and $t$, it is non-negative on $[-1,1]\times[-1,1]$. It is quadratic in \(t\): 
\begin{equation} 
\Psi((1,s), (t,1)) = (2 - 3/2 \cdot s^2)t^2 + (3s - 2)t + s^2 - 2s + \frac{3}{2}. 
\end{equation} 
Since \(4-3s^2\ge 1\), this quadratic is convex in \(t\) for any fixed $s \in [-1,1]$. If its minimum on \([-1,1]\) occurs at an endpoint, then we use 
\begin{equation} 
\Psi((1,s), (1,1)) = (3-s)(1+s)/2 \ge 0, \qquad \Psi((1,s), (-1,1)) = (1-s)(s+11)/2 \ge 0. 
\end{equation} 
Otherwise the minimum occurs at the vertex \(t_0=\frac{2-3s}{4-3s^2}\). The condition \(t_0\le 1\) implies \(s\ge (3-\sqrt{33})/6\), while \(t_0\ge -1\) is automatic for \(-1\le s\le 1\). At the vertex, 
\begin{equation} 
\Psi((1,s), (t_0,1)) = \frac{(1-s)(3s^3-3s^2+2s+4)}{4-3s^2}. 
\end{equation} 
The denominator is positive, and \(1-s\ge 0\). The polynomial \(p(s)=3s^3-3s^2+2s+4\) is increasing because \(p'(s)=9s^2-6s+2>0\). Hence, for \(s\ge (3-\sqrt{33})/6\), \begin{equation} p(s) \ge p\left(\frac{3-\sqrt{33}}{6}\right) = 6-\frac{2\sqrt{33}}{3} > 0. \end{equation} Thus \(\Psi((1,s), (t_0,1))\ge 0\), and therefore \(\Psi((1,s), (t,1))\ge 0\) on the whole square.

\item If \(x=(-1,s)\), we have \(\Psi(x,y) = \Psi((-1,s), (t,1))\), which is covered
by the same argument as in the previous point.
\end{enumerate}
\end{enumerate}
Therefore \(\Psi(x,y)\ge 0\) for all \(x,y\in [-1,1]^2\), proving inequality~\eqref{eq_goal_cf_t2} and concluding the proof.
\end{proof}

\section{Factorization costs for tree-based methods}
\label{appendix_table_analysis}
Table~\ref{tab:tree-methods} reports the $\MaxSE$ and $\MeanSE$ errors for several continual counting mechanisms, including smooth binary trees, $k$-ary trees with and without subtraction, and averaged $k$-ary trees with subtraction. Some of these mechanisms were not developed for the pure-DP setting and therefore did not explicitly state $\MaxSE$ and $\MeanSE$. For completeness, we provide calculations of these quantities for all mechanisms considered.

\subsection{Smooth binary trees}

Consider the factorization \(T_n=LR\) yielded by the smooth binary tree mechanism \cite{andersson2023smooth}. Take a tree of height \(h\), and
let the active leaves be the \(h\)-bit strings with exactly \(k\) zeros. The
construction uses \(k=h/2\).

The key combinatorial facts are the following. The sum values are stored only in the tree nodes that are left children. Then, each input coordinate contributes
to exactly \(k\) stored tree sums, so every column of \(R\) has \(\ell_1\)-norm
\(k\). Hence \(\|R\|_{1\to 1}=k\). Also, each prefix query is recovered as a
sum of exactly \(h-k\) stored tree sums, so every row of \(L\) has
\(h-k\) nonzero entries, all equal to \(1\). Therefore
\(\|L\|_{2\to\infty}=\sqrt{h-k}\), and \(\|L\|_F^2=n(h-k)\).

For the choice \(k=h/2\), we have $h$ such that ${h \choose h/2} \geq n$, implying $h = (1+o(1))\log_2 n$. Then
\begin{equation}
    \|L\|_{2\to\infty}\|R\|_{1\to 1}
=
\sqrt{h/2}\cdot h/2
=
\left(\frac{\log_2 n}{2}\right)^{3/2} (1+o(1)), \quad
\MaxSE = \frac{\log_2^3 n}{4}  (1+o(1)).
\end{equation}

Similarly, since \(\frac{1}{n}\|L\|_F^2\) coincides with \(\|L\|_{2\to\infty}^2\) for this factorization, we have
\begin{equation}
    \frac{1}{\sqrt{n}} \|L\|_{F}\|R\|_{1\to 1}
=
\left(\frac{\log_2 n}{2}\right)^{3/2} (1+o(1)), \quad
\MeanSE = \frac{\log_2^3 n}{4} (1+o(1)).
\end{equation}

\subsection{$k$-ary trees}

Let $LR = T_n$ be a decomposition corresponding to the $k$-ary tree mechanism. 

The depth of the tree is $h = (1+o(1)) \log_k n = \frac{\log_2 n}{\log_2 k} (1+o(1))$. Thus, for the right matrix, we have $\|R\|_{1 \to 1} = h = \frac{\log_2 n}{\log_2 k}(1 + o(1))$.

In the worst case, we need to sum about $(k-1) h = \frac{k-1}{\log_2 k} \log_2 n$ nodes from the tree to get a given prefix sum. Thus $\|L\|_{2 \to \infty} = \sqrt{\frac{k-1}{\log_2 k} \log_2 n} (1 + o(1))$. Then
\begin{equation}
    \|L\|_{2\to\infty}\|R\|_{1\to 1}
=
(1+o(1)) \frac{\sqrt{k-1}}{\log_2^{3/2} k}\log_2^{3/2} n , \quad
\MaxSE = (1+o(1)) \frac{2(k-1)}{\log_2^{3} k}\log_2^{3} n.
\end{equation}

The function $\frac{2(k-1)}{\log_2^{3} k}$ over integer $k$ has a minimum at $k=17$ with the value $\frac{32}{\log_2^{3} 17} \approx 0.4686$.

\subsection{Averaged $k$-ary trees with subtraction}
\label{sub:avaraged_k_ary}

In this subsection, we describe a new variant of the $k$-ary tree mechanism with subtraction, obtained by averaging two different decompositions of the prefix interval. This method appeared in the implementation code for the paper by Imola, Boninsegna, Keller, Aamand, Chowdhury, and Pagh~\cite{imola2026differentially} on private quantile estimation. However, it has not been described in that paper or elsewhere, except in the form of Python code. With the authors' permission, we provide a description and error analysis of the method.

Fix the value $k$. Let $h := \lceil \log_k n \rceil$. A complete $k$-ary tree of height $h$ is built on top of the input vector $x$. Every node of this tree holds a sum of a segment of the input vector, as in the usual $k$-ary tree mechanism. 

One can express the prefix sum $x_1 + \ldots + x_t$ as a suitable linear combination of the node values. Notice that
\begin{equation}
    [x_1, \ldots, x_t] = \frac{1}{2} \left([x_1, \ldots, x_n] + [x_1, \ldots, x_t] - [x_{t+1}, \ldots, x_n] \right),
    \label{eq_averaged_trees_subtraction_trick}
\end{equation}
where we mean addition and subtraction of intervals in terms of their indicator functions. Figures \ref{fig_3_ary_tree} and \ref{fig_3_ary_tree_h3_v2} illustrate this equality.

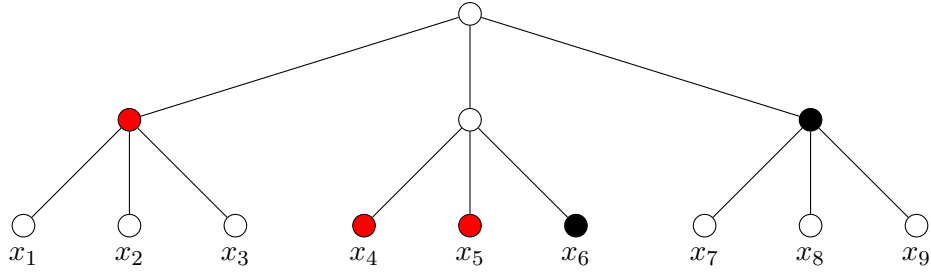
\begin{figure}[t]
    \centering
    \begin{tikzpicture}[
  level distance=1.4cm,
  level 1/.style={sibling distance=4.5cm},
  level 2/.style={sibling distance=1.4cm},
  vertex/.style={circle, draw, inner sep=3pt, minimum size=5pt},
  greenvertex/.style={vertex, fill=red!100},
  redvertex/.style={vertex, fill=black!100},
  leaf/.style={vertex},
  every child node/.style={vertex}
]
\node[vertex] {}
  child { node[greenvertex] {}
    child { node[leaf, label=below:{$x_1$}] {} }
    child { node[leaf, label=below:{$x_2$}] {} }
    child { node[leaf, label=below:{$x_3$}] {} }
  }
  child { node[vertex] {}
    child { node[greenvertex, label=below:{$x_4$}] {} }
    child { node[greenvertex, label=below:{$x_5$}] {} }
    child { node[redvertex, label=below:{$x_6$}] {} }
  }
  child { node[redvertex] {}
    child { node[leaf, label=below:{$x_7$}] {} }
    child { node[leaf, label=below:{$x_8$}] {} }
    child { node[leaf, label=below:{$x_9$}] {} }
  };
\end{tikzpicture}
    \caption{A 3-ary tree of height 2. The red nodes compute $[x_1, \ldots x_5]$, the black nodes compute $[x_6, \ldots x_9]$. Then $[x_1, \ldots x_5] = (\text{root} + \text{red} - \text{black})/2$.}
    \label{fig_3_ary_tree}
\end{figure}

\begin{figure}[!b]
    \centering
    \resizebox{\linewidth}{!}{%
\begin{tikzpicture}[
  level distance=1.35cm,
  level 1/.style={sibling distance=6cm},
  level 2/.style={sibling distance=2cm},
  level 3/.style={sibling distance=0.66cm},
  vertex/.style={circle, draw, inner sep=2pt, minimum size=5pt},
  greenvertex/.style={vertex, fill=red},
  redvertex/.style={vertex, fill=black},
  leaf/.style={vertex}
]

\node[vertex] {}
  child { node[greenvertex] {}
    child { node[vertex] {}
      child { node[leaf, label=below:{$x_1$}] {} }
      child { node[leaf, label=below:{$x_2$}] {} }
      child { node[leaf, label=below:{$x_3$}] {} }
    }
    child { node[vertex] {}
      child { node[leaf, label=below:{$x_4$}] {} }
      child { node[leaf, label=below:{$x_5$}] {} }
      child { node[leaf, label=below:{$x_6$}] {} }
    }
    child { node[vertex] {}
      child { node[leaf, label=below:{$x_7$}] {} }
      child { node[leaf, label=below:{$x_8$}] {} }
      child { node[leaf, label=below:{$x_9$}] {} }
    }
  }
  child { node[vertex] {}
    child { node[greenvertex] {}
      child { node[leaf, label=below:{$x_{10}$}] {} }
      child { node[leaf, label=below:{$x_{11}$}] {} }
      child { node[leaf, label=below:{$x_{12}$}] {} }
    }
    child { node[vertex] {}
      child { node[greenvertex, label=below:{$x_{13}$}] {} }
      child { node[greenvertex, label=below:{$x_{14}$}] {} }
      child { node[redvertex, label=below:{$x_{15}$}] {} }
    }
    child { node[redvertex] {}
      child { node[leaf, label=below:{$x_{16}$}] {} }
      child { node[leaf, label=below:{$x_{17}$}] {} }
      child { node[leaf, label=below:{$x_{18}$}] {} }
    }
  }
  child { node[redvertex] {}
    child { node[vertex] {}
      child { node[leaf, label=below:{$x_{19}$}] {} }
      child { node[leaf, label=below:{$x_{20}$}] {} }
      child { node[leaf, label=below:{$x_{21}$}] {} }
    }
    child { node[vertex] {}
      child { node[leaf, label=below:{$x_{22}$}] {} }
      child { node[leaf, label=below:{$x_{23}$}] {} }
      child { node[leaf, label=below:{$x_{24}$}] {} }
    }
    child { node[vertex] {}
      child { node[leaf, label=below:{$x_{25}$}] {} }
      child { node[leaf, label=below:{$x_{26}$}] {} }
      child { node[leaf, label=below:{$x_{27}$}] {} }
    }
  };

\end{tikzpicture}
}
    \caption{A 3-ary tree of height 3. The red nodes compute $[x_1, \ldots x_{14}]$, the black nodes compute $[x_{15}, \ldots x_{27}]$. Then $[x_1, \ldots x_{14}] = (\text{root} + \text{red} - \text{black})/2$.}
    \label{fig_3_ary_tree_h3_v2}
\end{figure}
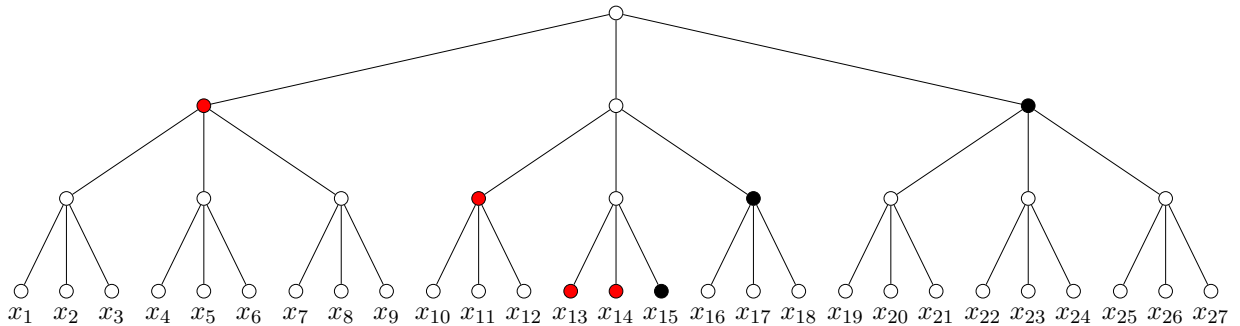

This gives the way to compute the prefix sum $[x_1, \ldots, x_t]$. First, compute $[x_1, \ldots, x_t]$ as in the usual $k$-ary tree (red nodes in Figure~\ref{fig_3_ary_tree}). Then compute $[x_{t+1}, \ldots, x_n]$ similarly (black nodes in Figure~\ref{fig_3_ary_tree}). Then express $[x_1, \ldots, x_t]$ as in Eq.~\eqref{eq_averaged_trees_subtraction_trick}: the root node plus the red nodes minus the black nodes, all divided by two. 

Now let us analyze the corresponding factorization costs. The right matrix is the usual $k$-ary tree matrix, with $\|R\|_{1 \to 1} = h + 1$. Consider the rows of the left matrix. In an easy case when the row index is $t=n=k^h$, this row has a single ``1'' corresponding to the root of the tree, and the 2-norm of the row is 1. Every other row has entries in $\{-1/2, 0, 1/2\}$. Suppose that colored vertices are present in the tree at levels $\{1, \ldots, h'\}$, where the root is considered level $0$. Notice that $h'$ does not have to coincide with $h$ if $t$ is divisible by $k$. Then every level in $\{1, \ldots, h'-1\}$ contains exactly $k-1$ colored nodes, and the level $h'$ contains $k$ colored nodes. Every row of $L$ uses red nodes, black nodes, and the root, and thus its 2-norm is at most $\frac{1}{2} \sqrt{(k-1)(h-1) + 1 + k} = \frac{1}{2} \sqrt{(k-1)h} \cdot (1+o(1))$. This gives us
\begin{equation}
    \|L\|_{2\to\infty}\|R\|_{1\to 1}
=
\frac{\sqrt{k-1} \cdot h^{3/2}}{2} (1+o(1)), \quad
\MaxSE = \frac{(k-1) \log_2^3 n}{2 \log_2^3 k}  (1+o(1)),
\end{equation}
\begin{equation}
    \frac{1}{\sqrt{n}} \|L\|_{F}\|R\|_{1\to 1}
=
\frac{\sqrt{k-1} \cdot h^{3/2}}{2} (1+o(1)), \quad
\MeanSE = \frac{(k-1) \log_2^3 n}{2 \log_2^3 k}  (1+o(1)).
\end{equation}

Notice that the values for $\MaxSE$ and $\MeanSE$ coincide. The best coefficient is attained at $k = 17$ and equals $\frac{8}{\log_2^3 17} \approx 0.1171$.

\end{document}